\documentclass[11pt]{article}
\usepackage[in]{fullpage}

\usepackage[shortlabels]{enumitem}
\usepackage{bm}
\usepackage{amsfonts}
\usepackage{amsmath}
\usepackage{amssymb}
\usepackage{xcolor}
\usepackage{mathtools}
\usepackage{graphicx}
\usepackage{amsthm}
\usepackage{qtree}
\usepackage{tree-dvips}
\usepackage{float}
\usepackage{braket}
\usepackage{mathrsfs}
\usepackage{tikz}
\usepackage{qcircuit}
\usepackage{dsfont}
\usepackage{authblk}
\usepackage{thm-restate}
\usepackage{longfbox}
\usepackage{comment}

\usepackage[style=alphabetic,minalphanames=3,maxalphanames=4,maxnames=99,backref=true]{biblatex}

\DeclareFieldFormat{eprint:iacr}{Cryptology ePrint Archive: \href{https://ia.cr/#1}{\texttt{#1}}}
\DeclareFieldFormat{eprint:iacrarchive}{Cryptology ePrint Archive: \href{https://eprint.iacr.org/archive/#1}{\texttt{#1}}}
\usepackage{hyperref}
\usepackage[nameinlink]{cleveref}
\hypersetup{colorlinks={true},linkcolor={blue},citecolor=magenta}

\theoremstyle{plain}
\newtheorem{theorem}{Theorem}[section]
\newtheorem{lemma}[theorem]{Lemma}
\newtheorem{claim}[theorem]{Claim}

\newtheorem{definition}[theorem]{Definition}
\newtheorem{protocol}{Protocol}
\crefname{step}{Step}{Steps}
\theoremstyle{remark}
\newtheorem{remark}[theorem]{Remark}

\newcommand{\ketbra}[2]{\ket{#1}\!\bra{#2}}
\newcommand{\proj}[1]{\ket{#1}\!\!\bra{#1}}

\newcommand{\R}{\mathbb{R}}
\newcommand{\C}{\mathbb{C}}
\newcommand{\N}{\mathbb{N}}

\newcommand{\Tr}{\mathrm{Tr}}

\newcommand{\reg}[1]{\mathsf{#1}}

\newcommand{\id}{\mathds{1}}

\newcommand{\eps}{\epsilon}

\newcommand{\fidelity}{\mathrm{F}}
\newcommand{\ot}{\otimes}
\newcommand{\sgn}{\mathrm{sgn}}

\newcommand{\Image}{\mathrm{Image}}
\newcommand{\Dom}{\mathrm{Dom}}

\newcommand{\class}[1]{\mathsf{#1}}
\newcommand{\poly}{\mathrm{poly}}

\newcommand{\td}{\mathrm{td}}
\newcommand*{\interact}{\mathord{\leftrightarrows}}

\newcommand{\avgUhlmann}{\textsc{DistUhlmann}}

\title{Local transformations of bipartite entanglement are rigid}

\author[1]{John Bostanci}
\author[2]{Tony Metger}
\author[1]{Henry Yuen}

\affil[1]{Columbia University}
\affil[2]{ETH Zurich}

\begin{document}

\date{}

\maketitle

\begin{abstract}
    Uhlmann's theorem is a fundamental result in quantum information theory that quantifies the optimal overlap between two bipartite pure states after applying local unitary operations (called \emph{Uhlmann transformations}). We 
    show that optimal Uhlmann transformations are \emph{rigid} -- in other words, they must be unique up to some well-characterized degrees of freedom. This rigidity is also \emph{robust}: Uhlmann transformations achieving near-optimal overlaps must be close to the unique optimal transformation (again, up to well-characterized degrees of freedom).  We describe two applications of our robust rigidity theorem: (a) we obtain better interactive proofs for synthesizing Uhlmann transformations and (b) we obtain a simple, alternative proof of the Gowers-Hatami theorem on the stability of approximate representations of finite groups.
\end{abstract}

\section{Introduction}
\label{sec:intro}

Let $\ket{C},\ket{D} \in \C^d \otimes \C^d$ denote bipartite pure states; let $\reg{A}$ denote the first subsystem and $\reg{B}$ denote the second subsystem. What is the closest that one can get to $\ket{D}$ by performing a unitary on subsystem $\reg{B}$ of the state $\ket{C}$? Uhlmann's theorem~\cite{uhlmann1976transition}  quantifies the optimal overlap achievable:
\begin{equation}
    \label{eq:intro-1}
        \fidelity(\rho,\sigma) = \max_U \, | \bra{D} \id \otimes U \ket{C}| \,,
\end{equation}
where $\rho$ and $\sigma$ denote the reduced density matrices on subsystem $\reg{A}$ of $\ket{C}$ and $\ket{D}$ respectively, the function $\fidelity(\rho,\sigma) = \Tr(\sqrt{ \rho^{1/2} \sigma \rho^{1/2}})$ denotes the fidelity between the two states, and the maximization is over all unitary transformations acting on subsystem $\reg{B}$. We call a unitary $U$ achieving equality in \Cref{eq:intro-1} an \emph{Uhlmann transformation}.

Given the ubiquity of Uhlmann's theorem throughout quantum information science, it seems worthwhile to study the mathematical and computational properties of Uhlmann transforms. Many natural questions arise: how unique are Uhlmann transformations? How robust are they to perturbations of the underlying states $\ket{C},\ket{D}$? What is the complexity of performing Uhlmann transformations on a quantum computer? Can difficult Uhlmann transformations be used for cryptography? The latter two questions were recently studied by Metger and Yuen~\cite{metger2023stateqip} and Bostanci, et al.~\cite{bostanci2023unitary}, who investigate a theory of state and unitary complexity, respectively. 

This paper studies the first two questions concerning the uniqueness and robustness of Uhlmann transformations. At first glance, Uhlmann transformations are not generally unique. For example, suppose that the reduced density matrix of $\ket{C}$ on subsystem $\reg{B}$ does not have full support. Then any Uhlmann transformation can behave arbitrarily on the orthogonal complement of the support, while remaining optimal. What if we disregard these trivial degrees of freedom, however -- could Uhlmann transformations be unique in some other meaningful way?

We provide an answer via \emph{canonical Uhlmann transformations}, first defined by Metger and Yuen~\cite{metger2023stateqip}. For every pair of bipartite states $(\ket{C},\ket{D}$), this is the operator
\begin{equation}
    \label{eq:canonical}
    W := \sgn( \Tr_{\reg{A}} (\ket{D}\!\!\bra{C}))
\end{equation}
where $\Tr_{\reg{A}}$ denotes tracing out register $\reg{A}$ and $\sgn(\cdot)$ denotes the following function: for a matrix $X$ with singular value decomposition $X = U \Sigma V^*$, we define $\sgn(X) \coloneqq U \sgn(\Sigma) V^*$ where $U,V$ are unitary operators and $\sgn(\Sigma)$ denotes the projection onto the eigenvectors of $\Sigma$ with positive eigenvalues (i.e., the support of $\Sigma$). The canonical transformation $W$ is a partial isometry\footnote{A partial isometry can be thought of as the restriction of a unitary to a subspace. More formally, an operator $W$ is a partial isometry if $W^* W$ is a projection.}; it is unitary if and only if both reduced states $\rho,\sigma$ (of $\ket{C}, \ket{D}$, respectively) are invertible.

The following was proven by Bostanci, et al.~\cite[Proposition 6.3]{bostanci2023unitary} in their investigation of the computational complexity of implementing Uhlmann transformations:
\begin{lemma}
    The canonical Uhlmann transformation $W$ satisfies $|\bra{D} \id \otimes W \ket{C}| = \fidelity(\rho,\sigma)$, and furthermore for all partial isometries $R$ such that $|\bra{D} \id \otimes R \ket{C}| = \fidelity(\rho,\sigma)$, we have that $W^* W \leq R^* R$ in the positive semidefinite ordering. 
\end{lemma}
In other words, the canonical Uhlmann transformation defined in~\Cref{eq:canonical} achieves the optimal overlap between $\ket{C}$ and $\ket{D}$, and furthermore any other partial isometry achieving the optimal overlap must be supported on the domain of $W$. This is a rather weak statement, however: when $R$ is unitary, then $W^* W \leq R^* R = \id$ is satisfied for \emph{all} partial isometries $W$. 

A stronger statement is the following:
\begin{claim}
\label{clm:exact-rigidity}
    For all partial isometries $R$ such that $\bra{D} \id \otimes R \ket{C} = \fidelity(\rho,\sigma)$, we have that
    \[
        \id \otimes W \ket{C} = \id \otimes R W^* W \ket{C}~.
    \]
\end{claim}
\noindent This says that \emph{any} optimal Uhlmann transformation, when restricted to the support of $W$, must behave identically to $W$ on the state $\ket{C}$. This provides some justification in calling the $W$ in \Cref{eq:canonical} the ``canonical'' Uhlmann transformation corresponding to $\ket{C},\ket{D}$. 

\Cref{clm:exact-rigidity} is in fact a special case of a more general \emph{robust rigidity} theorem that we prove in this paper. Roughly speaking, the theorem (\Cref{thm:main} below) states that any transformation $R$ that achieves \emph{approximately}-optimal fidelity (meaning that $|\bra{D} \id \otimes R \ket{C}| \geq \fidelity(\rho,\sigma) - \eps$) must be \emph{approximately} the canonical Uhlmann transformation $W$. This is analogous to rigidity results for some quantum information processing tasks such as nonlocal games~\cite{mayers2003self,mckague2012robust,vsupic2020self} and superdense coding~\cite{nayak2023rigidity}. These results show that the only way for a quantum operation to achieve near-optimal performance according to some metric (e.g., winning probability in a nonlocal game, or decoding probability in superdense coding) is if, in fact, it is close to a canonical strategy or protocol. 

First we give a way to quantify the rigidity of canonical Uhlmann transformations. 
\begin{definition}\label{def:robustness}
Let $\ket{C},\ket{D} \in \C^d \otimes \C^d$ be pure bipartite states with respective reduced density matrices $\rho,\sigma$ on the subsystem $\reg{A}$. Then we say the corresponding canonical Uhlmann transformation $W$ defined in \Cref{eq:canonical} has \emph{$\delta(\eps)$-robust rigidity} if, for all $\eps > 0$, for all unitaries $R$ such that 
    \[
        \bra{D} \id \ot R \ket{C} \geq \fidelity(\rho,\sigma) - \eps~,
    \]
    we have 
    \[
        \| \id \ot (W - R)W^* W \ket{C} \|^2 \leq \delta(\eps)~.
    \]
\end{definition}

\noindent Thus, bounds on the function $\delta(\eps)$ of an Uhlmann transformation quantifies the extent to which the exact rigidity statement of \Cref{clm:exact-rigidity} can be made robust. 

\begin{remark}
The reader may notice an apparent asymmetry between $\ket{C}$ and $\ket{D}$ in \Cref{clm:exact-rigidity} and \Cref{def:robustness}. This is motivated by an operational interpretation of Uhlmann's theorem: starting with $\ket{C}$, how close can we get to $\ket{D}$ by acting on subsystem $\reg{B}$? The choice of starting with $\ket{C}$ versus $\ket{D}$ is significant, as the canonical Uhlmann transformation can have different robustness functions depending on this choice (see \Cref{sec:eta-dependence} for an example). 
\end{remark}

\begin{remark}
The reader may also wonder about the role of $W^* W$ in \Cref{def:robustness}, which is the projection onto the image of $W$. The image of $W$ may not be fully contained in the support of subsystem $\reg{B}$ of $\ket{C}$ or $\ket{D}$. Interestingly, this projection is necessary in the statement of rigidity: \emph{any} unitary completion of the partial isometry $W$ achieves the optimal fidelity, as shown in \Cref{clm:completeness}. In other words, it is possible to behave arbitrarily outside the image of $W$, and still attain the optimal fidelity. 
\end{remark}

Our main result is a general bound on the rigidity of Uhlmann transformations.

\begin{restatable}[Robust rigidity of Uhlmann transformations]{theorem}{rigidity}
\label{thm:main}
    Let $\ket{C},\ket{D} \in \C^d \otimes \C^d$ be pure bipartite states with respective reduced density matrices $\rho,\sigma$ on the subsystem $\reg{A}$. The corresponding canonical Uhlmann transformation $W$ satisfies the following:
\begin{enumerate}
    \item (\textbf{Completeness}). For all unitary completions $U$ of $W$, we have
    \[
        | \bra{D} \id \ot U \ket{C} | = \fidelity(\rho,\sigma)~.
    \]
    \item (\textbf{Rigidity}). $W$ has $\delta(\eps)$-robust rigidity for $\delta(\eps) = \Big( \frac{2\kappa}{\eta} \Big) \eps$ where $\kappa=\| \rho^{-1/2} P \rho^{1/2} \|_\infty^2$ with $P$ being the projection onto $\Image(\rho^{1/2} \sigma \rho^{1/2})$, and $\eta$ is the smallest nonzero eigenvalue of the matrix geometric mean $\rho^{-1} \# \sigma$.
\end{enumerate}
\end{restatable}
\noindent For readers who are not familiar with the (beautiful notion of the) matrix geometric mean we provide a brief introduction in \Cref{sec:matrix-geometric-mean}. 

Thus, the canonical Uhlmann transformation is indeed robustly rigid, up to some blow-up that depends on two parameters $\eta$ and $\kappa$ (called the \emph{spectral gap} and \emph{obliqueness}, respectively) of the reduced density matrices $\rho,\sigma$. Intuitively, the obliqueness parameter $\kappa$ is a measure of a combination of non-commutativity and non-invertibility of $\rho,\sigma$ and the spectral gap parameter $\eta$ is a measure of how ``well-conditioned'' the matrix geometric mean $\rho^{-1} \# \sigma$ is (which one can think of as a notion of ``ratio'' between $\sigma$ and $\rho$).

\begin{remark}
    Suppose either
    \begin{enumerate}
        \item The density matrices $\rho,\sigma$ commute, or
        \item The density matrices $\rho,\sigma$ are invertible. 
    \end{enumerate}
    Then the obliqueness parameter $\kappa$ is equal to $1$, and the robustness bound only depends on the spectral gap $\eta$ of $\rho^{-1} \# \sigma$. 
\end{remark}

\noindent In \Cref{sec:eta-dependence,sec:kappa-dependence} respectively we further show that some dependence on the spectral gap parameter $\eta$ and the obliqueness parameter $\kappa$ in \Cref{thm:main} is necessary.

\subsection{Matrix geometric mean}
\label{sec:matrix-geometric-mean}

The matrix geometric mean is a noncommutative generalization of the geometric mean $\sqrt{ab}$ of two nonnegative numbers $a,b$. If $A,B$ are commuting positive semidefinite matrices, then the matrix geometric mean $A \# B$ is defined as $A^{1/2} B^{1/2}$. For general positive definite (i.e., all eigenvalues are strictly positive) matrices $A,B$, the matrix geometric mean $A \# B$ is defined as
\begin{equation}
    \label{eq:matrix-geometric-mean}
    A \# B := A^{1/2} (A^{-1/2} B A^{-1/2})^{1/2} A^{1/2}~.
\end{equation}
For positive definite matrices $A,B$ the matrix geometric mean enjoys many pleasant properties, including
\begin{enumerate}
    \item $A \# B$ is positive definite.
    \item $A\# B = B \# A$.
    \item $A \# B$ is the unique positive solution to the equation $X A^{-1} X = B$. 
    \item If $X$ is invertible, then $X(A\# B) X^{-1} = (XA X^{-1}) \# (X B X^{-1})$. 
    \item $A \# B \leq \frac{1}{2} ( A + B)$, a noncommutative analogue of the arithmetic-geometric mean inequality. 
    \item $\Phi(A) \# \Phi(B) \leq \Phi(A \# B)$ for all positive maps $\Phi$.
\end{enumerate}
Proofs of these properties can be found in~\cite[Chapter 4]{bhatia2009positive}. For more applications of the matrix geometric mean in quantum information theory, see~\cite{cree2020fidelity,fawzi2021defining,liu2024quantum}.

For noninvertible $A, B$ (but still positive \emph{semi}definite), the matrix geometric mean $A \# B$ is typically defined as a limit of geometric means of sequences of strictly positive matrices converging to $A,B$; however, in this case not all of the properties listed above are satisfied. For example, the symmetry property $A \# B = B \# A$ need not hold. 

In this paper we do not use this limit definition, and instead stick to \Cref{eq:matrix-geometric-mean} as the definition for the matrix geometric mean for \emph{all} positive semidefinite matrices $A, B$, with the inverses now being Moore-Penrose pseudoinverses.
Although it does not satisfy all the properties listed above, it satisfies a few important properties that are needed for the proof of \Cref{thm:main}; for example, as we will show in \Cref{clm:equiv-w}, when the density matrices $\rho,\sigma$ are real, the canonical Uhlmann transformation can equivalently be expressed in terms of a matrix geometric mean:
    \[
        W = Y^* (\rho^{1/2} \sigma^{1/2})^{-1} \rho^{1/2} (\rho^{-1} \# \sigma) \rho^{1/2} X
    \] 
    for some unitary operators $X,Y$ (see the start of \Cref{sec:rigidity} for their definitions). Furthermore, the fidelity between $\rho$ and $\sigma$ can also be written as
    $\fidelity(\rho,\sigma) = \Tr((\rho^{-1} \# \sigma) \rho)$ (see e.g.~\cite{cree2020fidelity}).

\subsection{Applications}

Given the centrality of Uhlmann transformations, the rigidity statement in~\Cref{thm:main} may be of interest in its own right, but it also turns out to be a useful technical tool for other applications.
To illustrate this, we briefly discuss applications of our robust rigidity theorem to  unitary complexity theory and approximate representation theory. 

\paragraph{The complexity of the Uhlmann Transformation Problem.} Bostanci, et al.~\cite{bostanci2023unitary} defined the Uhlmann Transformation Problem, a computational task associated to implementing canonical Uhlmann transformations corresponding to a pair ($\ket{C},\ket{D}$) whose circuit descriptions are given. They introduced a framework for unitary complexity theory in order to properly describe the complexity of performing Uhlmann transformations: for the special case that the pair $(\ket{C},\ket{D})$ have identical reduced density matrices (i.e., $\fidelity(\rho,\sigma) = 1$), the Uhlmann Transformation Problem is complete for $\mathsf{avgUnitaryHVPZK}$, a unitary complexity class that captures \emph{perfect zero knowledge} in the unitary synthesis setting~\cite[Theorem 6.1]{bostanci2023unitary}. They left open the challenge of characterizing the complexity of canonical Uhlmann transformations for general values of $\fidelity(\rho,\sigma)$. 

In \Cref{sec:complexity} we present a simple $2$-round quantum interactive synthesis protocol for the Uhlmann Transformation Problem (for all values of the fidelity of the reduced density matrices) -- this improves upon the $8$-round protocol that arises from the machinery of proving $\mathsf{avgUnitaryPSPACE} = \mathsf{avgUnitaryQIP}$ in~\cite{bostanci2023unitary}.  The soundness of our $2$-round protocol crucially depends on the robust Uhlmann rigidity theorem. We believe that this could be helpful for better understanding the complexity of the Uhlmann Transformation Problem in the future.

\paragraph{Approximate representation theory.} In the mathematics literature, results such as \Cref{thm:main} are known as \emph{stability} results: if an object $A$ approximately satisfies some constraints, then is it close (in the appropriate metric) to an object $B$ that \emph{exactly} satisfies those constraints~\cite{ulam1960collection}? In \Cref{sec:gowers-hatami}, we show that our robust Uhlmann rigidity theorem is powerful enough to derive other stability results -- in particular, we show that the Gowers-Hatami theorem on the stability of approximate representations of finite groups~\cite{gowers2015inverse} is an easy consequence of \Cref{thm:main}. Our proof suggests a possible ``mechanical template'' for proving other kinds of stability results: first, define the appropriate pair of pure states $(\ket{C},\ket{D})$, show that the canonical Uhlmann transformation is the ideal, ``exact'' object, and then use robust Uhlmann rigidity to conclude that all approximate objects are close to the ideal, exact object. We note that our approach of proving the Gowers-Hatami theorem is reminiscent of Metger, Natarajan, Zhang's alternate proof of it~\cite{metger2024succinct}.

\subsection{Related work}

\paragraph{A weaker rigidity theorem.} Bostanci, et al.~\cite{bostanci2023unitary} proved the following rigidity theorem for Uhlmann transformations, where the robustness itself depends on the fidelity between the reduced states:
\begin{theorem}[Weak Uhlmann rigidity]
\label{thm:weak-rigidity}
Let $\ket{C},\ket{D}$ be pure bipartite states with reduced density matrices $\rho,\sigma$ on the first subsystem. Then for all unitaries $R$ such that 
\[
    \bra{D} \id \otimes R \ket{C} \geq \fidelity(\rho,\sigma) - \eps~,
\]
we have that 
\[
    \| \id \otimes (W - R) \ket{C} \|^2 \leq 8 \Big( 1 - \fidelity(\rho,\sigma) + \sqrt{\eps} \Big)
\]
where $W$ is the corresponding canonical Uhlmann transformation.
\end{theorem}

\begin{remark}
    Technically, the theorem is stated in greater generality in~\cite{bostanci2023unitary} for arbitrary channels rather than unitaries; we specialize the theorem statement for this paper. 
\end{remark}

Suppose the fidelity $\fidelity(\rho,\sigma)$ is equal to $1$. In our language, \Cref{thm:weak-rigidity} implies that the canonical Uhlmann transformation $W$ has robust rigidity $\delta(\eps) \leq 8\sqrt{\eps}$. (At first glance, it may seem rather nice that there is no dependence on the spectral gap $\eta$ or the obliqueness parameter $\kappa$, but note that in this special case of $\fidelity(\rho,\sigma) = 1$, the two density matrices are identical and therefore $\eta = \kappa = 1$.) Thus \Cref{thm:weak-rigidity} implies a robust rigidity bound for the perfect fidelity setting.\footnote{We note that in the $\fidelity(\rho,\sigma) = 1$ case \Cref{thm:main} implies a quadratically-better robustness function $\delta(\eps) = 2\eps$.}

However, when $\fidelity(\rho,\sigma)$ is strictly less than $1$, then the rigidity bound of \Cref{thm:weak-rigidity} becomes trivial as $\eps \to 0$; the upper bound on the closeness of $R$ and $W$ is always at least $8(1 - \fidelity(\rho,\sigma))$, a quantity that is a constant compared to $\eps$. Furthermore this gives trivial upper bounds whenever $\fidelity(\rho,\sigma) \leq 7/8$. 

Our main theorem (\Cref{thm:main}), on the other hand, gives a nontrivial rigidity bound no matter what $\fidelity(\rho,\sigma)$ is. 

\paragraph{Rigidity in quantum information theory.} This paper is inspired by rigidity in nonlocal games (also known as \emph{self-testing} in the nonlocal game literature), which is the phenomenon that for many nonlocal games of interest (such as the CHSH game or the Magic Square game), near-optimal strategies must be close to a canonical optimal strategy~\cite{mayers2003self,mckague2012robust,miller2012optimal}. There is also a long line of work studying various aspects of rigidity in nonlocal games; we refer the reader to the extensive survey of~\cite{vsupic2020self}. Nonlocal game rigidity is a powerful tool in quantum cryptography and quantum complexity theory, with applications ranging from classical verification of quantum computations~\cite{reichardt2013classical} to settling the complexity of quantum multiprover interactive proofs~\cite{ji2021mip}.

\paragraph{Stability of polar decompositions.} The rigidity of Uhlmann transformations is loosely related to the \emph{stability of polar decompositions}, a topic that has been studied extensively in numerical analysis~\cite{higham1986computing}. Every square matrix $A$ admits a polar decomposition $A = UP$ where $U$ (the ``polar factor'' of $A$) is a partial isometry and $P$ is a positive semidefinite matrix. How do the polar factors of a matrix $A$ and a perturbation $A + \Delta A$ compare with each other, as a function of $A$ and the perturbation $\Delta A$? This is a central question to the study of numerical algorithms for computing the polar decomposition. 

The connection with the Uhlmann transformation is as follows. The canonical Uhlmann transformation $W$ for a pair $(\ket{C},\ket{D})$ of states with corresponding density matrices $\rho,\sigma$ can be derived from the polar decomposition of the matrix $A = \sqrt{\rho} \sqrt{\sigma}$. Perturbing the states $\rho,\sigma$ (and consequently the states $\ket{C},\ket{D}$) will perturb the canonical Uhlmann transformation $W$; this relationship is governed by the stability of the polar decomposition of $A$.\footnote{See \Cref{sec:sensitivity} for an illustration of how the canonical Uhlmann transformation is a sensitive function of the states $\ket{C},\ket{D}$.}

However, the robust rigidity of Uhlmann transformations studied in this paper is a different notion of stability. Here, we do not consider perturbations of the states $\ket{C},\ket{D}$; we are asking whether all \emph{approximate} Uhlmann transformations $R$ for a pair of states $(\ket{C},\ket{D})$ must be close to a unique \emph{exact} Uhlmann transformation.

\subsection{Summary}

Uhlmann transformations, which are local transformations of bipartite (pure state) entanglement, are fundamental in quantum information theory. In this paper we showed that Uhlmann transformations possess a robust form of \emph{rigidity}: near-optimal entanglement transformations must close (in a well-defined sense) to a unique optimal transformation. This unique optimal transformation is the canonical Uhlmann transformation introduced by~\cite{metger2023stateqip}, and our result gives further justification to calling it ``canonical.'' 

We showed that the robustness of the rigidity theorem inherently depends on two parameters called the spectral gap and obliqueness, which are functions of the underlying pure states $\ket{C},\ket{D}$. An interesting open question is whether there is a general ``rounding'' procedure that converts any pair of states $(\ket{C},\ket{D})$ into a nearby pair $(\ket{\tilde{C}},\ket{\tilde{D}})$ with controlled spectral gap and controlled obliqueness. In \Cref{sec:rounding}, we provide a rounding lemma that only controls the spectral gap. 

Finally, we presented two applications of our robust rigidity theorem.  The first is to unitary complexity theory, where we can improve the round complexity required for the task of synthesizing canonical Uhlmann transformations in the regime where the fidelity between the reduced states is not $1$.  Second, we demonstrate that the robust rigidity theorem is a very general form of robustness that can be used to derive other stability theorems.  As an example, we re-derive the stability of approximate group representations by reducing to the robust rigidity of the canonical Uhlmann transformation between a specific pair of states.  Given the ubiquity of Uhlmann transformations in protocols across quantum information theory and quantum complexity theory, developing an understanding of their rigidity properties should unlock further applications and deeper insight into the ways that bipartite entanglement can be locally transformed. 

\section{Notation and facts about states and matrices}
\label{sec:prelims}

For a square matrix $A$,
\begin{enumerate}
    \item $A^{-1}$ denotes its Moore-Penrose pseudoinverse,
    \item $A^*$ denotes its conjugate transpose,
    \item $\overline{A}$ denotes its entrywise complex conjugate (with respect to the standard basis),
    \item $A^\top$ denotes its transpose (with respect to the standard basis).
    \item $\Image(A)$ denotes its image, i.e., $\mathrm{span} \{ A \ket{v} \}$,
    \item $\Dom(A)$ denotes its domain, i.e., the orthogonal complement of the kernel of $A$, 
    \item $\| A \|_1$ denotes its Schatten-$1$ norm, i.e., $\Tr(\sqrt{A^* A})$,
    \item $\|A\|_\infty$ denotes its operator norm, i.e., its largest singular value.
\end{enumerate}
Note that $\Image(A^* A) = \Dom(A) = \Image(A^*)$. 

Throughout we let $\ket{\Omega} = \sum_{i = 1}^d \ket{i} \otimes \ket{i}$ denote the unnormalized maximally entangled state on $\C^d \otimes \C^d$. We recall the well-known ``reflection'' property of maximally entangled states:
\begin{claim}
    \label{clm:reflection-property}
    For all operators $A$ acting on $\C^d$, we have
    \[
        A \otimes \id \ket{\Omega} = \id \otimes A^\top \ket{\Omega}~.
    \]
\end{claim}

Next we express every bipartite pure state in terms of the maximally entangled state. 
\begin{claim}
\label{clm:pure-state}
    Every pure bipartite state $\ket{C} \in \C^d \otimes \C^d$ whose reduced density matrix on the first register is $\rho$ can be written as 
    \[
        \ket{C} = \sqrt{\rho} \otimes X  \ket{\Omega}
    \]
    for some unitary operator $X$.
\end{claim}
\begin{proof}
    Let $\ket{C} = \sum_{i = 1}^d \lambda_i \ket{a_i} \otimes \ket{b_i}$ denote a Schmidt decomposition of $\ket{C}$, so the reduced state is $\rho = \sum_i \lambda_i^2 \ketbra{a_i}{a_i}$. Let $X = \sum_{i=1}^d \ketbra{b_i}{\overline{a_i}}$ where $\ket{\overline{a_i}}$ denotes the complex conjugate of $\ket{a_i}$. The claim follows since the unnormalized maximally entangled state $\ket{\Omega}$ can be equivalently expressed as 
    \[
        \ket{\Omega} = \sum_{i=1}^d \ket{a_i} \otimes \ket{\overline{a_i}}~.
    \]
\end{proof}

\begin{lemma}[Schur complement lemma]
\label{lem:schur}
Let $M$ be the following block matrix:
\[
    M = \begin{pmatrix} A & B \\ B^* & C \end{pmatrix}~.
\]
Then $M$ is positive semidefinite if and only if $A \geq 0$ and $(\id - AA^{-1}) B = 0$ and $C \geq B^* A^{-1} B$. 
\end{lemma}
\noindent Proofs of this can be found in, e.g.,~\cite[Appendix A.5.5]{boyd2004convex} and~\cite[Chapter 1]{bhatia2009positive}.

\subsection{Semidefinite programming}

A \emph{semidefinite program (SDP) in standard form} is specified by self-adjoint matrices $B \in \C^{m \times m},C^{n \times n}$, and a Hermiticity-preserving\footnote{This means that if $X$ is self-adjoint, then so is $\Phi(X)$.} superoperator $\Phi$ mapping operators on $\C^n$ to operators on $\C^m$, and corresponds to the following optimization problem:
\begin{equation*}
\begin{array}{ll@{}ll}
\underset{X}{\text{maximize}}  & \Tr(CX) &\\
\\
\text{subject to}& \displaystyle 
\Phi(X) = B  \\
\\
& \displaystyle X \geq 0
\end{array}
\label{eq:primal-sdp-prelims}
\end{equation*}
where $X$ ranges over all positive semidefinite operators on $\C^n$.

The \emph{dual SDP} of the standard-form program above is
\begin{equation*}
\begin{array}{ll@{}ll}
\text{minimize}  & \Tr(BY) &\\
\\
\text{subject to}& \displaystyle 
\Phi^*(Y) \geq C  \\
\\
& \displaystyle Y \text{ Hermitian}
\end{array}
\end{equation*}
where $\Phi^*$ denotes the \emph{adjoint} of the superoperator $\Phi$ (meaning that $\Phi^*$ is the unique superoperator satisfying $\langle Y, \Phi(X) \rangle = \langle \Phi^*(Y), X \rangle$ for all operators $X \in \C^{n \times n},Y \in \C^{m \times m}$). 

\emph{Weak SDP duality} states that the objective value of the dual SDP is always an upper bound on the objective value of the primal SDP. We refer the reader to Watrous's textbook for a detailed treatment of SDPs in the context of quantum information~\cite{watrous2018theory}.

\section{Rigidity of Uhlmann transforms} \label{sec:rigidity}

We now prove the main theorem of the paper, \Cref{thm:main}, which for convenience we restate here.
\rigidity*
Let $\ket{C},\ket{D}$ be bipartite pure states. For notational convenience, we let $\rho,\sigma$ denote the entry-wise \emph{complex conjugates} of the reduced density matrices of $\ket{C},\ket{D}$ on the subsystem $\reg{A}$ respectively -- it is easy to verify that $\fidelity(\overline{\rho},\overline{\sigma}) = \fidelity(\rho,\sigma)$. Then by \Cref{clm:pure-state}, we can write
\begin{align*}
    \ket{C} &= \sqrt{\overline{\rho}} \otimes X \ket{\Omega} \\
    \ket{D} &= \sqrt{\overline{\sigma}} \otimes Y \ket{\Omega} 
\end{align*}
for some unitary operators $X,Y$. 

For simplicity, we prove \Cref{thm:main} for the special case that $X = Y = \id$. This is without loss of generality: if \Cref{thm:main} holds for a pair $\ket{C} = \sqrt{\overline{\rho}} \otimes \id \ket{\Omega}, \ket{D} = \sqrt{\overline{\sigma}} \otimes \id \ket{\Omega}$ with canonical Uhlmann transformation $W$, then it is easy to see that \Cref{thm:main} holds also for $\ket{C} = \sqrt{\overline{\rho}} \otimes X \ket{\Omega}, \ket{D} = \sqrt{\overline{\sigma}} \otimes Y \ket{\Omega}$ with canonical Uhlmann transformation $\widetilde{W} = YWX^*$. 

Thus, for the rest of this section we assume $X = Y = \id$. 

\subsection{Properties of the canonical Uhlmann transformation}

First, we establish equivalent expressions for the canonical Uhlmann transformation in terms of the reduced density matrices $\rho,\sigma$.

\begin{claim}
\label{clm:equiv-w}
    When $X = Y = \id$, the following expressions are equivalent definitions of the canonical Uhlmann transformation $W$ for the pair $(\ket{C},\ket{D}$):
    \begin{itemize}
        \item $ W = \sgn(\sqrt{\sigma} \sqrt{\rho})$~.
        \item $W = (\rho^{1/2} \sigma^{1/2})^{-1} \rho^{1/2} (\rho^{-1} \# \sigma) \rho^{1/2}$.
    \end{itemize}
\end{claim}
\begin{proof}
From the reflection property of maximally entangled states, we have
\begin{gather*}
    \ket{C} = \id \otimes  (\sqrt{\overline{\rho}})^\top \ket{\Omega} \\
    \ket{D} = \id \otimes  (\sqrt{\overline{\sigma}})^\top \ket{\Omega}~.
\end{gather*}
On the other hand, since $\rho$ is positive semidefinite, we have
\[
    (\sqrt{\overline{\rho}})^\top = \sqrt{\overline{\rho}^\top} = \sqrt{\rho^*} = \sqrt{\rho}
\]
and similarly $(\sqrt{\overline{\sigma}})^\top = \sqrt{\sigma}$.
Thus, recalling the definition of $W$ from \Cref{eq:canonical}, we have
    \begin{align*}
        W &= \sgn(\Tr_{\reg{A}}(\ketbra{D}{C})) \\
        &= \sgn(\Tr_{\reg{A}}((\id  \otimes \sqrt{\sigma}) \ketbra{\Omega}{\Omega} (\id \otimes \sqrt{\rho} ))) \\
&= \sgn(\sqrt{\sigma} \sqrt{\rho})
    \end{align*}
    where the third line uses the fact that the partial trace of $\ketbra{\Omega}{\Omega}$ is the identity matrix.

    Let $U\Sigma V^*$ denote the singular value decomposition of the product $\sqrt{\sigma} \sqrt{\rho}$. Then $\sgn(\sqrt{\sigma} \sqrt{\rho}) = U \sgn(\Sigma) V^*$ by definition. On the other hand,
    \begin{align*}
        (\rho^{1/2} \sigma^{1/2})^{-1} \rho^{1/2} (\rho^{-1} \# \sigma) \rho^{1/2} &= (\rho^{1/2} \sigma^{1/2})^{-1} (\rho^{1/2} \sigma \rho^{1/2})^{1/2} \\
        &= (V \Sigma U^*)^{-1} (V\Sigma^2 V^*)^{1/2} \\
        &= U \Sigma^{-1} V^* V \Sigma V^* \\
        &= U \sgn(\Sigma) V^* \\
        &= \sgn(\sqrt{\sigma} \sqrt{\rho})
    \end{align*}
    where the first line uses that 
    \[
        \rho^{1/2} (\rho^{-1} \# \sigma) \rho^{1/2} = \rho^{1/2} \rho^{-1/2} (\rho^{1/2} \sigma \rho^{1/2})^{1/2} \rho^{-1/2} \rho^{1/2} = (\rho^{1/2} \sigma \rho^{1/2})^{1/2}~.
    \]
    To see this, we use two facts:
    \begin{enumerate}
        \item $\rho^{1/2} \rho^{-1/2}$ is the projection onto $\Image(\rho)$
        \item $\Image \Big((\rho^{1/2} \sigma \rho^{1/2})^{1/2} \Big) = \Image \Big(\rho^{1/2} \sigma \rho^{1/2} \Big) \subseteq \Image(\rho)$.
    \end{enumerate}
    Therefore $\rho^{1/2} \rho^{-1/2}$ acts as the identity on $(\rho^{1/2} \sigma \rho^{1/2})^{1/2}$.
    This concludes the proof of the claim. 
\end{proof}

\begin{claim}
\label{clm:basic-properties-w} The following hold:
\begin{enumerate}
    \item $WW^*$ is the projection onto $\Image(\sigma^{1/2} \rho \sigma^{1/2})$, and
    \item $W^* W$ is the projection onto $\Image(\rho^{1/2} \sigma \rho^{1/2})$.
\end{enumerate}
\end{claim}
\begin{proof}
These items follow from the proof of \Cref{clm:equiv-w}, which established that $W = \sgn(\sqrt{\sigma} \sqrt{\rho})$. Let $\sqrt{\sigma} \sqrt{\rho}$ have singular value decomposition $U \Sigma V^*$. Then $W = U \sgn(\Sigma) V^*$. Therefore
\[
    WW^* = U \sgn(\Sigma) U^*
\]
which is the projection onto the image of $\sigma^{1/2} \rho \sigma^{1/2} = U \Sigma U^*$. Similarly, $W^* W = V \sgn(\Sigma) V^*$ which is the projection onto the image of $\rho^{1/2} \sigma \rho^{1/2} = V \Sigma V^*$.

\end{proof}
We verify the completeness property of the canonical Uhlmann transformation. 
\begin{claim} \label{clm:completeness} For all unitary completions $U$ of $W$, we have
$\bra{D} \id \ot U \ket{C} = \fidelity(\rho,\sigma)$. 
\end{claim}
\begin{proof}
First, observe that 
\begin{align*}
     \bra{D} \id \ot W \ket{C} = \Tr(\sigma^{1/2} W \rho^{1/2}) = \Tr(\rho^{1/2} \sigma^{1/2} \, \sgn(\sigma^{1/2} \rho^{1/2})) = \Tr((\rho^{1/2} \sigma \rho^{1/2})^{1/2}) = \fidelity(\rho,\sigma)
\end{align*}
where we used the fact that for all square matrices $A$, $\Tr(A \, \sgn(A^*)) = \Tr(\sqrt{AA^*})$ and we used the definition of the fidelity function.

Now let $U = W + E$ be a unitary completion of $W$. Suppose for sake of contradiction that $\bra{D} \id \otimes U \ket{C} \neq \fidelity(\rho,\sigma)$. Then this implies that $\bra{D} \id \otimes E \ket{C} \neq 0$. Let $e^{i \theta}$ be a complex phase such that $e^{i \theta} \bra{D} \id \otimes E \ket{C}$ is a strictly positive number. Then consider the unitary $U' = W + e^{i \theta} E$. Then
\[
    \bra{D} \id \otimes U' \ket{C} = \fidelity(\rho,\sigma) + e^{i \theta} \bra{D} \id \otimes E \ket{C} > \fidelity(\rho,\sigma)
\]
which contradicts Uhlmann's theorem. 
\end{proof} 

\subsection{The rigidity proof}

Now we prove that the canonical Uhlmann transformation $W$ has $\delta(\eps)$-robust rigidity for $\delta(\eps) = (2\kappa/\eta) \eps$. We do this by setting up a semidefinite program (SDP) whose objective is to maximize the distance
\[
\| \id \ot W \ket{C} - \id \ot R W^* W \ket{C} \|^2
\]
when ranging over all unitary operators $R$ satisfying $\bra{D} \id \ot R \ket{C} \geq \fidelity(\rho,\sigma) - \eps$. We analyze the dual SDP, which is a minimization problem, and provide a feasible solution whose objective is $2\kappa \eps/\eta$, which by (weak) SDP duality gives an upper bound on the objective value of the primal SDP. 

For convenience, we present some helpful notation for the rest of the proof:
\begin{itemize}
    \item $W$ is the canonical Uhlmann transformation corresponding to $(\ket{C},\ket{D})$,
    \item $A = \sqrt{\sigma} \sqrt{\rho}$,
    \item $P = W^* W$ is the projection onto $\Image(\rho^{1/2} \sigma \rho^{1/2})$ by \Cref{clm:basic-properties-w}. 
\end{itemize}
The robustness function $\delta(\eps)$ of $W$ can be cast as the optimum of the following optimization problem (P), where $R$ ranges over all matrices.
\begin{equation}
\label{eqn:og_optimization}
\tag{P}
\begin{array}{ll@{}ll}
\underset{R}{\text{maximize}}  & \| \id \ot (W - R)P \ket{C} \|^2 &\\
\\
\text{subject to}& \displaystyle 
\bra{D} \id \ot R \ket{C} \geq \fidelity(\rho,\sigma) - \eps  \\
\\
& \displaystyle R^* R = \id \,.
\end{array} 
\end{equation}
The next claim shows how to recast this optimization problem as a standard-form SDP.

\begin{claim} \label{cl:std_sdp}
Let $\omega$ denote the value of the following SDP:
\begin{equation}
\tag{Primal SDP}
\begin{array}{ll@{}ll}
\underset{X}{\text{maximize}}  & \Tr(CX) &\\
\\
\text{subject to}& \displaystyle 
\Phi(X) = B  \\
\\
& \displaystyle X \geq 0
\end{array}
\label{eqn:primal-sdp}
\end{equation}
where $X$ is a block-matrix $\begin{pmatrix} X_1 & R \\ R^* & X_2 \end{pmatrix}$, the matrix $C = \frac{1}{2} \begin{pmatrix} 0 & - W \, \rho P\\ 
- P \rho \, W^*   & 0 \end{pmatrix}$, $\Phi$ is the Hermiticity-preserving superoperator
\[ 
    \Phi(X) = \begin{pmatrix} X_1 &  &\\  & X_2 &  \\ & & \frac{1}{2} \Tr(R A^*) + \frac{1}{2} \Tr(R^* A)  
    \end{pmatrix} \,,
\]
where $A = \sigma^{1/2} \rho^{1/2}$, and $B = \begin{pmatrix} \id & & \\ & \id & \\ & & \fidelity(\rho,\sigma) - \eps \end{pmatrix}$. 
Then $2(\omega + \Tr(P\rho))$ is an upper bound to the optimal value of the optimization problem (\ref{eqn:og_optimization}).
\end{claim}
\begin{proof}

We first relax the constraint on $R^* R = \id$ in (\ref{eqn:og_optimization}) to $R^* R \leq \id$; an upper bound for this relaxed optimization problem is also an upper bound for (\ref{eqn:og_optimization}).
Using \Cref{lem:schur} we observe that $R^* R \leq \id$ is equivalent to $\begin{pmatrix} \id & R \\ R^* & \id \end{pmatrix} \geq 0$.

Next, we observe that we can replace the condition $\bra{D} \id \ot R \ket{C}  \geq \fidelity(\rho,\sigma) - \eps$ in (\ref{eqn:og_optimization}) by $\bra{D} \id \ot R \ket{C}  = \fidelity(\rho,\sigma) - \eps$, i.e., the maximizer in (\ref{eqn:og_optimization}) will always satisfy this condition with equality.
This is because any $R$ that achieves a better-than-necessary fidelity can be perturbed in a way that lowers the fidelity, but does not decrease the distance $\| \id \ot (W - R)P \ket{C} \|^2$ from the optimal Uhlmann transform.
Secondly, the condition $\bra{D} \id \ot R \ket{C}  = \fidelity(\rho,\sigma) - \eps$ implicitly imposes that $R$ is chosen such that $\bra{D} \id \ot R \ket{C}$ is real; we can relax this condition to $\Re \bra{D} \id \ot R \ket{C}  = \fidelity(\rho,\sigma) - \eps$, or equivalently
\begin{align*}
\frac{1}{2} \Tr(R A^*) + \frac{1}{2} \Tr(R^* A) = \fidelity(\rho,\sigma) - \eps \,,
\end{align*}
where $A = \sigma^{1/2} \rho^{1/2}$.

Third, we turn to the objective function in (\ref{eqn:og_optimization}).
Using that $P = W^* W$ is a projection and $P W^* = W^*$, as well as the condition $R^* R \leq \id$, we get
\begin{align*}
\| \id \ot (W - R)P \ket{C} \|^2 &= \bra{C} \id \ot P W^* W P \ket{C} + \bra{C} \id \ot R R^* R P \ket{C} - 2 \Re \bra{C} P W^* R P \ket{C} \\
&\leq 2 \; \Tr (P \rho) - \big( \Tr(R P \rho \, W^*) + \Tr( W \rho\, P R^* )  \big) \,.
\end{align*}
Thus we see that the optimization problem (\ref{eqn:og_optimization}) is upper-bounded the following SDP, except the objective value is scaled by $2$ and shifted by a constant $2\Tr(P \rho)$:
\begin{equation*}
\begin{array}{ll@{}ll}
\text{maximize}  & - \frac{1}{2} \Tr(R P \rho \, W^*) - \frac{1}{2} \Tr( W \rho\, P R^* ) &\\
\\
\text{subject to}& \displaystyle 
\begin{pmatrix} \id & R \\ R^* & \id \end{pmatrix} \geq 0  \\
\\
& \displaystyle \frac{1}{2} \Tr(R A^*) + \frac{1}{2} \Tr(R^* A) = \fidelity(\rho,\sigma) - \eps~.
\end{array}
\end{equation*}
Putting this SDP into standard form yields \Cref{cl:std_sdp}.
\end{proof}

We now analyze the dual of (\ref{eqn:primal-sdp}). 
\begin{claim} \label{cl:sdp_dual}
The following is the SDP dual of (\ref{eqn:primal-sdp}):
\begin{align}
\tag{Dual SDP}
\begin{array}{ll@{}ll}
\underset{Y}{\text{minimize}}  & \Tr(Y_1) + \Tr(Y_2) + \alpha (\fidelity(\rho,\sigma) - \eps)     &\\
\\
\text{subject to}& \displaystyle 
\begin{pmatrix} Y_1 & \frac{1}{2} \alpha A  \\ \frac{1}{2} \alpha A^* & Y_2 \end{pmatrix} \geq \frac{1}{2} \begin{pmatrix} & -  W \, \rho P\\ -P \rho \, W^* & \end{pmatrix}   \\
\\
& \displaystyle Y \text{ Hermitian}~.
\end{array}
\label{eqn:dual-sdp}
\end{align}
where $Y$ ranges over all block matrices of the form $\begin{pmatrix} Y_1 & * & * \\ * & Y_2 & * \\ * & * & \alpha \end{pmatrix}$ with $\alpha \in \R$. 
\end{claim}
\begin{proof}
Taking the dual of (\ref{eqn:primal-sdp}) we get
\begin{equation*}
\begin{array}{ll@{}ll}
\text{minimize}  & \Tr(BY) &\\
\\
\text{subject to}& \displaystyle 
\Phi^*(Y) \geq C  \\
\\
& \displaystyle Y \text{ Hermitian}
\end{array}
\end{equation*}
We compute the adjoint $\Phi^*$. For all block matrices $Y = \begin{pmatrix} Y_1 & * & * \\ * & Y_2 & * \\ * & * & \alpha \end{pmatrix}$ where $\alpha$ is a scalar, we claim that $\Phi^*$ evaluated on $Y$ is given by
\[
    \Phi^*(Y) = \begin{pmatrix} Y_1 & \frac{1}{2} \alpha A \\ \frac{1}{2} \alpha A^* & Y_2 \end{pmatrix}~.
\]
We can verify that this is the adjoint because 
\[
    \Tr(Y \Phi(X)) = \Tr(Y_1 X_1) + \Tr(Y_2 X_2) + \frac{1}{2} \alpha \Tr(R A^*) + \frac{1}{2} \alpha \Tr(R^* A)  = \Tr(\Phi^*(Y) X)~. 
\]
It then follows from weak SDP duality that an upper bound on the value of the dual SDP is also an upper bound on the value of the primal SDP.
\end{proof}

We now consider feasible certificates for (\ref{eqn:dual-sdp}), which give \emph{upper bounds} on the value of (\ref{eqn:dual-sdp}), which by weak SDP duality gives upper bounds on the value of (\ref{eqn:primal-sdp}).

\begin{claim} \label{cl:choose_Y}
For all $\alpha \in \mathbb{R}$, the matrix 
\begin{align*}
Y = \begin{pmatrix}
    \sqrt{T^* T} && \\
    & T (\sqrt{T^* T})^{-1} T^* & \\
    && \alpha
\end{pmatrix}
\end{align*}
is feasible for (\ref{eqn:dual-sdp}), where $T = \frac{1}{2} \left( \alpha A^* + P \rho W^* \right)$, and the corresponding value of the objective function is
\[
2 \| T \|_1 + \alpha (\fidelity(\rho,\sigma) - \eps)~.
\]
\end{claim}
\begin{proof}
Our choice of $Y$ is manifestly self-adjoint, so we only need to check that it satisfies the constraint 
\begin{align*}
\begin{pmatrix} Y_1 & \frac{1}{2} \alpha A  \\ \frac{1}{2} \alpha A^* & Y_2 \end{pmatrix} \geq \frac{1}{2} \begin{pmatrix} & -  W \, \rho PW\\ -P \rho \, W^* & \end{pmatrix}  \,.
\end{align*}
For our choice of $Y$, this is equivalent to 
\begin{align*}
\begin{pmatrix} \sqrt{T^* T} & T^*  \\ T & T (\sqrt{T^* T})^{-1} T^* \end{pmatrix} \geq 0 \,.
\end{align*}
This follows easily from \Cref{lem:schur}.
The only non-trivial condition to check is that $(\id - \sqrt{T^* T} (\sqrt{T^* T})^{-1}) T^* = 0$.
This holds because $\sqrt{T^* T} (\sqrt{T^* T})^{-1}$ is the projector onto $\Dom(\sqrt{T^* T}) = \Dom(T^* T) = \Dom(T) = \Image(T^*)$.

The value achieved by this $Y$ is 
\begin{align*}
\Tr(Y_1) + \Tr(Y_2) + \alpha (\fidelity(\rho,\sigma) - \eps) 
&= \Tr(\sqrt{T^* T}) + \Tr(T (\sqrt{T^* T})^{-1} T^*) + \alpha (\fidelity(\rho,\sigma) - \eps) \\
&= 2 \lVert{T}\rVert_1 + \alpha (\fidelity(\rho,\sigma) - \eps) \,,
\end{align*}
where the last line holds because $\Tr(T (\sqrt{T^* T})^{-1} T^*) = \Tr(\sqrt{T^* T} (\sqrt{T^* T})^{-1} \sqrt{T^* T}) = \Tr(\sqrt{T^* T}) = \lVert{T}\rVert_1$ by the properties of the pseudoinverse.
\end{proof}

Recall that $\kappa = \| \rho^{-1/2} P \rho^{1/2} \|_\infty^2$ and $\eta$ is the smallest nonzero eigenvalue of $\rho^{-1} \# \sigma$. The next claim computes a bound on the value of the dual SDP for $\alpha = -\kappa/\eta$. 

\begin{claim}
\label{clm:dual-sdp-bound}
Let $\alpha = -\kappa/\eta$. Then for the feasible solution $Y$ given by \Cref{cl:choose_Y} achieves value $(\kappa/\eta) \eps - \Tr(P \rho)$.
\end{claim}
\begin{proof}
By definition of $T$, we have
\[ 
2 \| T\|_1 = \|  \alpha A^* + P \rho W^* \|_1 = \|  \alpha A^* W + P \rho P \|_1~.
    \]
This is because we can write $W = Q U$ for some projection $Q$ and a unitary $U$. The projection $Q$ is equal to $W W^*$, which by \Cref{clm:basic-properties-w} is the projection onto $\Image(\sigma^{1/2} \rho \sigma^{1/2})$, or equivalently onto $\Dom(\rho^{1/2} \sigma^{1/2})$. Thus $A^* = \rho^{1/2} \sigma^{1/2} Q = A^* Q$ and $W^* = W^* Q$. The equality then follows by unitary invariance of the trace norm.

For our choice of $\alpha$, $2\| T \|_1$ is equal to
\[
    \| P \rho P - (\kappa/\eta) A^* W \|_1 = \| (\kappa/\eta) A^* W - P \rho P \|_1~.
\]
We will argue that 
\begin{equation}
\label{eq:psd}
    (\kappa/\eta) A^* W - P \rho P \geq 0
\end{equation}
in the positive semidefinite ordering. This will complete the proof of the claim because then the objective value of the dual SDP is equal to
\begin{align*}
    \| (\kappa/\eta) A^* W - P \rho P \|_1 + \alpha (\fidelity(\rho,\sigma) - \eps) &= \Tr( (\kappa/\eta) A^* W - P \rho P ) + \alpha (\fidelity(\rho,\sigma) - \eps) \\
    &= (\kappa/\eta)  \Tr( A^* W) - \Tr( P \rho) - (\kappa/\eta) \fidelity(\rho,\sigma) + (\kappa/\eta) \eps \\
    &= (\kappa/\eta) \eps - \Tr(P \rho)
\end{align*}
as desired. The last line follows because $\Tr(A^* W) = \fidelity(\rho,\sigma)$ as shown in the proof of \Cref{clm:completeness}.

We now prove \Cref{eq:psd}. The proof of \Cref{clm:equiv-w} shows that $W = (\rho^{1/2} \sigma^{1/2})^{-1} (\rho^{1/2} \sigma \rho^{1/2})^{1/2}$. Thus 
\[
    A^* W = \rho^{1/2} (\rho^{-1} \# \sigma) \rho^{1/2}~,
\]
because $A^* (A^*)^{-1}$ is equal to the projection onto $\Image(\rho^{1/2} \sigma^{1/2})$. 
From this expression we see that $A^* W$ is self-adjoint.

We now argue that
\begin{equation}
    \label{eq:dom-incl}
    \Dom(\rho^{-1/2} P \rho P \rho^{-1/2}) \subseteq \Dom(\rho^{-1} \# \sigma)~.
\end{equation}
This is equivalent to showing that
\[
     \ker(\rho^{-1/2} P \rho P \rho^{-1/2}) \supseteq \ker(\rho^{-1} \# \sigma)~.
\]
Let $\ket{v}$ be a vector in $ \ker(\rho^{-1} \# \sigma)$, which implies
\[
    \bra{v} \rho^{-1/2} (\rho^{1/2} \sigma \rho^{1/2})^{1/2} \rho^{-1/2} \ket{v} = 0~.
\]
This implies that the vector $\rho^{-1/2} \ket{v} \in \ker( (\rho^{1/2} \sigma \rho^{1/2})^{1/2})$. On the other hand, $\ker( (\rho^{1/2} \sigma \rho^{1/2})^{1/2}) = \ker( \rho^{1/2} \sigma \rho^{1/2})$, so therefore $P \rho^{-1/2} \ket{v} = 0$. This implies that $\ket{v} \in \ker(\rho^{-1/2} P \rho P \rho^{-1/2})$, as desired. 

Let $\Pi_{\Dom(\rho^{-1/2} P \rho P \rho^{-1/2})}$ and $\Pi_{\Dom(\rho^{-1} \# \sigma)}$ denote the projections onto the specified domains. Thus
\begin{align*}
    \rho^{-1/2} P \rho P \rho^{-1/2} &\leq \kappa \,  \Pi_{\Dom(\rho^{-1/2} P \rho P \rho^{-1/2})} \\
    &\leq \kappa \, \Pi_{\Dom(\rho^{-1} \# \sigma)} \\
    &\leq \frac{\kappa}{\eta} \rho^{-1} \# \sigma~.
\end{align*}
The first line follows from the definition of $\kappa$, the second line follows from \Cref{eq:dom-incl}, and the third line follows because $\eta$ is the smallest non-zero eigenvalue of $\rho^{-1} \# \sigma$. 

To conclude, we can conjugate by $\rho^{1/2}$ (which preserves the operator inequality) to get 
\[
    P \rho P \leq \frac{\kappa}{\eta} \rho^{1/2} (\rho^{-1} \# \sigma) \rho^{1/2} = \frac{\kappa}{\eta} A^* W
\]
as desired. This uses the fact that $\rho^{1/2} \rho^{-1/2} P = P \rho^{-1/2} \rho^{1/2} = P$, as shown in the proof of \Cref{clm:equiv-w}. 
\end{proof}

Combining \Cref{clm:dual-sdp-bound} with \Cref{cl:std_sdp} shows that the optimal value of (\ref{eqn:og_optimization}) is at most
\[
    2(\omega + \Tr(P \rho)) \leq \Big( \frac{2\kappa}{\eta} \Big) \, \eps
\]
as desired. This concludes the proof of the rigidity property of \Cref{thm:main}.

\section{Sensitivity of Uhlmann transformations}
\label{sec:sensitivity}

In this section we present several examples that illustrate the \emph{sensitivity} of the canonical Uhlmann transformation. Some of these examples will demonstrate that it is necessary for the rigidity bound of \Cref{thm:main} to depend on both the spectral gap $\eta$ as well as the obliqueness $\kappa$.

\subsection{Canonical Uhlmann transformation is not a smooth function of input states}

The canonical Uhlmann transformation $W$ is not a smoothly-varying function of the states $\ket{C},\ket{D}$. It is instructive to consider the following two-qutrit example:
\begin{align*}
\ket{C} &= \sqrt{1 - \eps} \ket{00} + \sqrt{\eps/2} \ket{11} + 
\sqrt{\eps/2} \ket{22}\,, \\
\ket{\tilde{C}} &= \sqrt{1 - \eps} \ket{00} + \sqrt{\eps/2} \ket{12} + 
\sqrt{\eps/2} \ket{21}\,, \\
\ket{D} &= \ket{C}\,.
\end{align*}
The canonical Uhlmann transformation $W$ corresponding to $(\ket{C},\ket{D})$ is simply the identity operator on $\C^3$. On the other hand, the Uhlmann transformation $\tilde{W}$ corresponding to $(\ket{\tilde{C}},\ket{D})$ can be computed as
\[
    \tilde{W} = \ketbra{0}{0} + \ketbra{1}{2} + \ketbra{2}{1}\,.
\]
In other words, it swaps $\ket{1}$ with $\ket{2}$ and keeps $\ket{0}$ unchanged. The difference $W - \tilde{W}$ has operator norm at least $2$, but the difference $\ket{C} - \ket{\tilde{C}}$ has norm at most $\eps$, which can be arbitrarily small.

\subsection{The dependence on the spectral gap parameter $\eta$}
\label{sec:eta-dependence}

The parameter $\eta$ in \Cref{thm:main} can be viewed as quantifying how ``well-conditioned'' the pair of states $(\ket{C},\ket{D})$ are. We now present an example of a pair of states $(\ket{C},\ket{D})$ that 
\begin{enumerate}
   \item Shows the dependence on the spectral gap $\eta$ in the rigidity bound of \Cref{thm:main} is necessary, and 
   \item Demonstrates a scenario where the canonical Uhlmann transformation $W$ corresponding to the pair $(\ket{C},\ket{D}$) has a different robustness than the Uhlmann transformation $W^*$ corresponding to the flipped pair $(\ket{D},\ket{C})$.
\end{enumerate}

\begin{lemma}
\label{lem:eta-dependence}
    For all even $d \in \N$ and for all $\eta > 0$, there exists a pair $(\ket{C},\ket{D})$ of $d$-dimensional states with reduced density matrices $\rho,\sigma$ respectively such that
 \begin{enumerate}
        \item $\fidelity(\rho,\sigma) \geq \frac{1}{2}$,
        \item The smallest nonzero eigenvalue of $\rho^{-1} \# \sigma$ is $\eta$. 
        
        \item Let $W$ denote the canonical Uhlmann transformation for $(\ket{C},\ket{D})$. For all $0 < \eps < 1$ there exists a unitary $R$ such that
    \[
    \bra{D} \id \otimes R \ket{C} \geq \fidelity(\rho,\sigma) - \eps
    \]
    but 
    \[
        \| \id \otimes (W - R)W^* W \ket{C} \|^2 \geq 2 \eps/\eta~.
    \]
    \item Let $W^*$ denote the canonical Uhlmann transformation for $(\ket{D},\ket{C})$. 
    For all unitaries $Q$ such that
    \[
        \bra{C} \id \otimes Q \ket{D} \geq \fidelity(\rho,\sigma) - \eps
    \]
    we have
    \[
        \| \id \otimes (W^* - Q) WW^* \ket{D} \|^2 \leq 2\sqrt{2}\eps~.
    \]
    \end{enumerate}
\end{lemma}

\begin{proof}
Fix $\eta$ and let $\delta = \eta^2/2$. Let $A = \{1,2,\ldots,d/2\}$ and $B = \{d/2 + 1,\ldots,d\}$.
Consider the pair of states
\begin{gather*}
    \ket{C} = \frac{1}{\sqrt{d}} \sum_{i=1}^d \ket{i} \otimes \ket{i}~, \qquad \qquad \ket{D} = \sum_{i=1}^d \sqrt{\sigma_i} \ket{i} \otimes \ket{i}
\end{gather*}
where $\sigma_i = 2(1 - \delta)/d$ if $i \in A$ and $\sigma_i = 2\delta/d$ otherwise. Let $\sigma$ be the diagonal matrix with entries $(\sigma_i)$, and let $\rho = \id/d$ denote the maximally mixed state. By construction $\rho,\sigma$ are the reduced density matrices of $\ket{C},\ket{D}$ on register $\reg{A}$, respectively. 

The fidelity between $\rho,\sigma$ is 
\[
    \fidelity(\rho,\sigma) = \sum_i \sqrt{\sigma_i/d} = \frac{d}{2} \Big( \sqrt{2(1 - \delta)/d^2}  + \sqrt{2\delta/d^2} \Big) = \sqrt{\frac{1 - \delta}{2}} + \sqrt{\frac{\delta}{2}} \geq \frac{1}{2}~.
\]

The canonical Uhlmann transformation $W$ for $(\ket{C},\ket{D})$ is the following:
\[
    W = \sgn(\Tr_{\reg{A}}(\ketbra{D}{C})) = \sgn \Big(  \sum_i \sqrt{\frac{\sigma_i}{d}} \ketbra{i}{i} \Big) = \id~.
\]
Here we used that $\sigma_i > 0$ for all $i$. Similarly, the canonical Uhlmann transformation $W^*$ for $(\ket{D},\ket{C})$ is also the identity matrix. However their robustnesses are different.

\paragraph{Robustness of $W$.} The matrix geometric mean $\rho^{-1} \# \sigma$ is a diagonal matrix whose entries are $\sqrt{d \sigma_i}$. The smallest nonzero eigenvalue is $\sqrt{2\delta} = \eta$. Since $\rho,\sigma$ commute, the obliqueness parameter $\kappa$ is equal to $1$. 

Let $0 < \tau < 1$. Let $G \subseteq B$ be a subset such that $|G|/|B| \geq \tau$. Consider the following Uhlmann transformation $R$: it acts as identity on the basis states $\ket{i}$ for $i \notin G$, and applies a \emph{derangement} (i.e., a permutation without any fixed points) on the basis states in $G$. The performance of this Uhlmann transformation is
\[
    \bra{D} \id \otimes R \ket{C} = \sum_{i \notin G} \sqrt{\rho_i \sigma_i} \geq \frac{d}{2} \sqrt{\frac{2(1- \delta)}{d^2}} + \frac{(1 - \eps)d}{2} \sqrt{\frac{2\delta}{d^2}} = \sqrt{\frac{1 - \delta}{2}} + (1 - \tau) \sqrt{\frac{\delta}{2}}
\]
Thus, the fidelity deficit of $R$ is $\eps = \tau \sqrt{\delta/2}$. On the other hand, we can evaluate the distance of $R$ from $W = \id$ (when evaluated on $\ket{C}$ and using the fact that $W^* W = \id$):
\[
    \| \ket{C} - \id \otimes R \ket{C} \|^2 = 2 \Big \| \sum_{i \in G} \frac{1}{d} \Big \|^2 = 2|G|/d \geq \tau~.
\]
This exactly matches the robust rigidity upper bound given by \Cref{thm:main}, which is
\[
  \delta(\eps) = (2\kappa/\eta) \eps = \frac{2}{\sqrt{2\delta}} \sqrt{\frac{\delta}{2}} \tau = \tau~.
\]
Thus, when the parameter $\eta$ is very small, there exist near-optimal Uhlmann transformations that can be far away from the canonical one, when measured on $\ket{C}$. 

\paragraph{Robustness of $W^*$.} Now we investigate the robustness of the canonical Uhlmann transformation for the opposite state transformation (going from $\ket{D}$ to $\ket{C}$). 
The matrix geometric mean $\sigma^{-1} \# \rho$ has diagonal entries $\frac{1}{\sqrt{d \sigma_i}}$, and thus the smallest nonzero eigenvalue is $\eta' = \frac{1}{\sqrt{2(1 - \delta)}} \geq \frac{1}{\sqrt{2}}$. The obliqueness parameter $\kappa'$ is still $1$. 

Thus, by \Cref{thm:main}, any unitary $Q$ such that
\[
    \bra{C} \id \otimes Q \ket{D} \geq \fidelity(\rho,\sigma) - \eps
\]
must satisfy
\[
    \| \ket{D} - \id \otimes Q \ket{D} \|^2 \leq (2/\eta') \eps \leq 2\sqrt{2} \, \eps~.
\]
Therefore the rigidity of the \emph{opposite} state transformation (from $\ket{D}$ to $\ket{C}$) is much more robust when measured on $\ket{D}$.

\end{proof}

\subsection{A rounding lemma for the spectral gap}
\label{sec:rounding}
\Cref{lem:eta-dependence} shows that the dependence on the spectral gap $\eta$ in the rigidity of Uhlmann transformations is unavoidable. In the following, we show that every pair of states $(\ket{C},\ket{D})$, which might have a very small spectral gap $\eta$, is \emph{close} to another pair of states $(\ket{\tilde{C}},\ket{\tilde{D}})$ with a better spectral gap $\tilde{\eta}$. We call this a \emph{rounding lemma} for the spectral gap parameter. 

\begin{restatable}[Rounding lemma for the spectral gap]{lemma}{closeness}
\label{lem:closeness}
    For all $0 < \eta < 1$, for all pure bipartite states $\ket{C},\ket{D} \in \C^d \otimes \C^d$ with reduced density matrices $\rho,\sigma$ on subsystem $\reg{A}$, there exist  bipartite states $\ket{\tilde{C}}, \ket{\tilde{D}}$ with reduced density matrices $\tilde{\rho},\tilde{\sigma}$ such that 
    \begin{enumerate}
        \item The smallest nonzero eigenvalue of $ \tilde{\rho}^{-1} \# \tilde{\sigma}$ is at least $\eta$, and
        \item Both $| \langle \tilde{D} | D \rangle|^2$ and $\langle \tilde{C} | C \rangle |^2$ are at least $1 - \eta^2$.
    \end{enumerate}
\end{restatable}

\begin{proof}
Let $\ket{C},\ket{D}$ be two states such that
\begin{gather*}
    \ket{C} = \sqrt{\rho} \otimes X \ket{\Omega}~, \qquad \qquad \ket{D} = \sqrt{\sigma} \otimes Y \ket{\Omega}~.
\end{gather*}
First we argue that $\ket{C},\ket{D}$ are arbitrarily close to states $\ket{\hat{C}},\ket{\hat{D}}$ where the corresponding reduced density matrices $\hat{\rho},\hat{\sigma}$ are invertible.
Define
\[
    \hat{\rho} = (1 - \delta) \rho + \delta \frac{\id}{d}, \qquad \qquad \hat{\sigma}  = (1 - \delta) \sigma + \delta \frac{\id}{d}
\]
and define
\[
    \ket{\hat{C}} = \sqrt{\hat{\rho}} \ot X \ket{\Omega}, \qquad \qquad \ket{\hat{D}} = \sqrt{\hat{\sigma}} \ot Y \ket{\Omega}~.
\]
Clearly, $|\langle \hat{C} | C \rangle |^2 = \Tr(\sqrt{\hat{\rho}} \sqrt{\rho})^2 \geq 1 - \delta$ where we used the fact that $\sqrt{\hat{\rho}} \geq \sqrt{(1 - \delta)} \sqrt{\rho}$ in the positive semidefinite ordering. Similarly, $|\langle \hat{D} | D \rangle |^2 \geq 1 - \delta$.  Note that for any nonzero value of $\delta$, the density matrices $\hat{\sigma},\hat{\rho}$ are invertible. 

Thus we assume without loss of generality that $\rho,\sigma$ are invertible to begin with. We now show there exists a density matrix $\tilde{\sigma}$ that is $O(\eta)$-close to $\sigma$ such the least nonzero eigenvalue of $\rho^{-1} \# \tilde{\sigma}$ is at least $\eta$. 

Let $\Pi$ denote the projection onto the eigenspace of $\rho^{-1} \# \sigma$ with eigenvalues at least $\eta$. Then 
\[
    \Tr( (\id - \Pi) (\rho^{-1} \# \sigma) \rho) \leq \eta~.
\]
On the other hand, by direct calculation we have that for invertible $\rho, \sigma$,
\[
    (\rho^{-1} \# \sigma) \rho = (\rho^{-1} \# \sigma) \rho (\rho^{-1} \# \sigma) (\rho \# \sigma^{-1}) = \sigma (\rho \# \sigma^{-1})~.
\]
Next, $\id - \Pi$ projects onto the eigenspace of $\rho \# \sigma^{-1}$ with eigenvalues at least $1/\eta$. Therefore
\[
    \Tr( (\id - \Pi) (\rho^{-1} \# \sigma) \rho) = \Tr(\sigma (\rho \# \sigma^{-1}) (\id - \Pi)) \geq \frac{1}{\eta} \Tr(\sigma (\id - \Pi))~.
\]
Thus $\Tr(\sigma (\id - \Pi) ) \leq \eta^2$. 
Let $\tilde{\sigma} = \Pi \sigma \Pi/\Tr(\Pi \sigma)$. 
By the Gentle Measurement Lemma (see~\cite[Corollary 3.15]{watrous2018theory}), $\fidelity(\sigma, \tilde{\sigma}) \geq \sqrt{1 - \eta^2}$. We now evaluate
\begin{align*}
     \rho^{-1} \# \tilde{\sigma} = \beta \Big ( \, ( \rho^{-1} \# \Pi \sigma \Pi) \Big ) = \beta \rho^{-1/2} \Big( \rho^{1/2} \Pi \sigma \Pi \rho^{1/2} \Big)^{1/2} \rho^{-1/2}~.
\end{align*}

\begin{claim}
    Suppose $B,X$ are positive semidefinite matrices and let $A$ be strictly positive. Then if
    \[
        \begin{pmatrix} A & X \\ X & B \end{pmatrix} \geq 0
    \]
    then for all projections $\Pi$ that commute with $X$ we have 
    \[
        A^{1/2} (A^{-1/2} \Pi B \Pi A^{-1/2})^{1/2} A^{1/2} \geq \Pi X \Pi~.
    \]
\end{claim}
\begin{proof}
By conjugating with $\begin{pmatrix} \id & 0 \\ 0 & \Pi \end{pmatrix}$ we get
\[
    \begin{pmatrix} A & \Pi X \Pi \\ \Pi X \Pi & \Pi B \Pi \end{pmatrix} = \begin{pmatrix} A & X \Pi \\ \Pi X & \Pi B \Pi \end{pmatrix} \geq 0~.
\]
    Here we used that $\Pi X = X \Pi = \Pi X \Pi$ since $\Pi$ is a projector and $X$ and $\Pi$ commute by assumption. By the Schur complement lemma (\Cref{lem:schur}) we have
    \[
        \Pi B \Pi \geq \Pi X \Pi A^{-1} \Pi X \Pi ~.
    \]
     Conjugating by $A^{-1/2}$ and using operator monotonicity of the square root we have
    \[
        (A^{-1/2} \Pi B \Pi A^{-1/2})^{1/2} \geq A^{-1/2} \Pi X \Pi A^{-1/2}~.
    \]
    Since $A$ is invertible we can multiply both sides by $A^{1/2}$ and get
    \[
        A^{1/2} (A^{-1/2} \Pi B \Pi A^{-1/2})^{1/2} A^{1/2} \geq \Pi X \Pi~.
    \]
\end{proof}

Using this preceding Claim with $A = \rho^{-1}, B = \sigma, X = \rho^{-1} \# \sigma$ and $\Pi = \Pi$ (which commutes with $X$) we get 
\[
    \rho^{-1} \# \tilde{\sigma} \geq \beta \Pi (\rho^{-1} \# \sigma) \Pi~.
\]
Therefore $\rho^{-1} \#\tilde{\sigma}$ has a larger spectral gap than $\Pi (\rho^{-1} \# \sigma) \Pi$ since $\beta$ is at least $1$. 

\end{proof}

\subsection{The dependence on the obliqueness parameter $\kappa$}
\label{sec:kappa-dependence}

We now show that the dependence on the parameter $\kappa$ is also necessary in the rigidity bound of \Cref{thm:main}. Recall that $P$ is the projection onto the image of $\rho^{1/2} \sigma \rho^{1/2}$. The intuition behind the obliqueness parameter $\kappa = \| \rho^{-1/2} P \rho^{1/2} \|_\infty^2$ is less clear than that of $\eta$ (at least, to us). 

First, note that the operator $\rho^{-1/2} P \rho^{1/2}$ is a projection operator, in that 
\[
    (\rho^{-1/2} P \rho^{1/2})^2 = \rho^{-1/2} P \rho^{1/2}~.
\]
However, it is not necessarily self-adjoint, in which case it is called an \emph{oblique} projection. The spectral norm of oblique projections are always at least one, but can be much larger.

One can think of $\kappa$ as measuring some combination of how noncommuting and noninvertible $\rho$ and $\sigma$ are; if either (a) $\rho,\sigma$ commute or (b) $\rho,\sigma$ are invertible, then $\kappa$ is always $1$ (like in the example constructed in \Cref{lem:eta-dependence}). On the other hand, we can construct examples of pairs of states where the spectral norm parameter $\kappa$ can be arbitrarily large, and the robustness of the canonical Uhlmann transformation must depend on $\kappa$ and $\eps$ (albeit in a way that doesn't yet match the upper bounds given by \Cref{thm:main}). 

\begin{lemma}
\label{lem:kappa-dependence}
    For all $d \in \N$, for all $\kappa \geq 1$, and for all $0 < \eps \leq \kappa^{-1/2}$, there exists a pair $(\ket{C},\ket{D})$ of $d$-dimensional states with reduced density matrices $\rho,\sigma$ respectively such that 
    \begin{enumerate}
        \item The parameter $\kappa$ satisfies $\kappa = \| \rho^{-1/2} P \rho^{1/2} \|_\infty^2$ where $P$ is the projection onto the image of $\rho^{1/2} \sigma \rho^{1/2}$,
        \item The smallest nonzero eigenvalue $\eta$ of $\rho^{-1} \# \sigma$ is at least $1$,
        \item There exists a unitary $R$ and an $\eps$ (depending on $\eta$) such that
    \[
    \bra{D} \id \otimes R \ket{C} = \fidelity(\rho,\sigma) - \eps
    \]
    but
    \[
        \| (\id \otimes (W - R)W^* W \ket{C} \|^2 \geq \kappa \, \eps^2~.
    \]
    \end{enumerate}
\end{lemma}
\begin{proof}
Let $\rho$ be an invertible density matrix, and let $\sigma = \ketbra{\sigma}{\sigma}$ be an arbitrary pure state. First, we calculate the corresponding quantities $\kappa$ and $\eta$ for these density matrices.

\begin{claim}
    Let $P$ denote the projection onto the image of $\rho^{1/2} \sigma \rho^{1/2}$. We have that
    \begin{enumerate}
        \item $\kappa = \| \rho^{-1/2} P \rho^{1/2} \|^2_\infty = \frac{\bra{\sigma} \rho^2 \ket{\sigma}}{\bra{\sigma} \rho \ket{\sigma}^2}$
        \item The smallest nonzero eigenvalue $\eta$ of $\rho^{-1} \# \sigma$ is $\frac{1}{\sqrt{\bra{\sigma} \rho \ket{\sigma}}}$.
    \end{enumerate}
\end{claim}
\begin{proof}
    First we compute
    \[
    \rho^{1/2} \sigma \rho^{1/2} = \frac{\rho^{1/2} \ketbra{\sigma}{\sigma} \rho^{1/2}}{\bra{\sigma} \rho \ket{\sigma}} \bra{\sigma} \rho \ket{\sigma}~.
    \]
    The operator $P$ projects onto this pure state, so we can write the projection as
    \[
    P = \frac{\rho^{1/2} \ketbra{\sigma}{\sigma} \rho^{1/2}}{\bra{\sigma} \rho \ket{\sigma}}~.
    \]
    We now compute the $\kappa$ parameter, which is the operator norm of the following matrix
    \[
    \rho^{-1/2} P \rho P \rho^{-1/2} = \ketbra{\sigma}{\sigma} \cdot \frac{\bra{\sigma} \rho^2 \ket{\sigma}}{\bra{\sigma} \rho \ket{\sigma}^2}~.
    \]
    
    For pure $\sigma = \proj{\sigma}$, we have that for any vector $\ket{\alpha}$, 
    \begin{align*}
    \rho^{1/2} \sigma \rho^{1/2} \ket{\alpha} = \bra{\sigma} \rho^{1/2} \ket{\alpha} \cdot \rho^{1/2} \ket{\sigma} \,,
    \end{align*}
    so $\Image(\rho^{1/2} \sigma \rho^{1/2})$ is spanned solely by the vector $\rho^{1/2} \ket{\sigma}$. Hence, $P$ is the normalized projector onto this vector: 
    \begin{align*}
    P = \frac{\rho^{1/2} \ketbra{\sigma}{\sigma} \rho^{1/2}}{\bra{\sigma} \rho \ket{\sigma}}~.
    \end{align*}
    On the other hand, the $\eta$ parameter can be computed as follows. The matrix geometric mean $\rho^{-1} \# \sigma$ can be computed as
    \begin{align*}
        \rho^{-1/2} (\rho^{1/2} \sigma \rho^{1/2})^{1/2} \rho^{-1/2} = \frac{\ketbra{\sigma}{\sigma}}{\sqrt{\bra{\sigma} \rho \ket{\sigma}}}~.
    \end{align*}
    Thus the smallest nonzero eigenvalue is $\frac{1}{\sqrt{\bra{\sigma} \rho \ket{\sigma}}}$.
\end{proof}

Define the pure states
\[
    \ket{C} = \id \otimes \sqrt{\rho} \ket{\Omega}~, \qquad \ket{D} = \ket{\overline{\sigma}} \otimes \ket{\sigma}
\]
where $\ket{\overline{\sigma}}$ denotes the entry-wise complex conjugate of $\ket{\sigma}$ with respect to the standard basis. Note that $\overline{\rho},\overline{\sigma}$ are the reduced density matrices of $\ket{C},\ket{D}$ on the first register, respectively, where $\overline{\rho},\overline{\sigma}$ denote the entry-wise complex conjugate of the density matrices $\rho,\sigma$, respectively.

The canonical Uhlmann transformation $W$ for the pair $(\ket{C},\ket{D})$ is
\[
    W = \sgn(\Tr_{\reg{A}}(\ketbra{D}{C})) = \ketbra{\sigma}{v}
\]
where $\ket{v} = \rho^{1/2} \ket{\sigma}/\sqrt{\bra{\sigma} \rho \ket{\sigma}}$. We can verify that this achieves the Uhlmann fidelity between the states $\ket{C},\ket{D}$:
\[
    \bra{D} \id \otimes W \ket{C} = \bra{v} \rho^{1/2} \ket{\sigma} = \sqrt{\bra{\sigma} \rho \ket{\sigma}} = \fidelity(\rho,\sigma) = \fidelity(\overline{\rho},\overline{\sigma})~.
\]  

We now analyze the extent to which the canonical Uhlmann transformation $W$ is rigid. Let $\ket{\hat{v}}$ be a state such that $\braket{v | \hat{v}} = 1 - \eps/\fidelity(\rho,\sigma)$, and let $\ket{\hat{v}} = \ket{v} + \ket{\tau}$ for some subnormalized vector $\ket{\tau}$. Consider the following partial isometry
\[
    R = \ketbra{\sigma}{\hat{v}} + \ketbra{\sigma^\perp}{\hat{v}^\perp}
\]
where $\ket{\hat{v}},\ket{\hat{v}^\perp}$ are orthogonal and $\ket{\sigma},\ket{\sigma^\perp}$ are orthogonal. Then 
\begin{align*}
    \bra{D} \id \otimes R \ket{C} &= \bra{\hat{v}} \rho^{1/2} \ket{\sigma} = \bra{v} \rho^{1/2} \ket{\sigma} + \bra{\tau} \rho^{1/2} \ket{\sigma} = \fidelity(\rho,\sigma) ( 1 + \braket{\tau | v} ) = 
    \fidelity(\rho,\sigma) - \eps~.
\end{align*}
On the other hand, we have
\begin{gather*}
    \id \otimes W \ket{C} = \sqrt{\overline{\rho}} \ket{\overline{v}} \otimes \ket{\sigma} \\
    \id \otimes R W^* W \ket{C} = (\id \otimes R \ketbra{v}{v})(\sqrt{\overline{\rho}} \otimes \id) \ket{\Omega} = \sqrt{\overline{\rho}} \ket{\overline{v}} \otimes R \ket{v}
\end{gather*}
where $\ket{\overline{v}}$ denotes the complex conjugate of $\ket{v}$. 
Now,
\[
R \ket{v} = \ket{\sigma} \braket{\hat{v} | v} + \ket{\sigma^\perp} \braket{\hat{v}^\perp | v}
\]
so that
\begin{align*}
    \Big \| \id \otimes (W - RW^* W) \ket{C} \Big \| &= \Big \| \ket{\sigma} (1 - \braket{\hat{v} | v}) - \ket{\sigma^\perp} \braket{\hat{v}^\perp | v} \Big \| \cdot \Big \| \sqrt{\overline{\rho}} \ket{\overline{v}} \Big \| \\
    &= \sqrt{ |\braket{\tau | v}|^2 + |\braket{\hat{v}^\perp | v}|^2} \cdot \sqrt{\frac{\bra{\sigma} \rho^2 \ket{\sigma}}{\bra{\sigma} \rho \ket{\sigma}^2}} \cdot \sqrt{\bra{\sigma} \rho \ket{\sigma}} \\
    &\geq \fidelity(\rho,\sigma) \cdot |\braket{\tau | v}| \cdot \sqrt{\frac{\bra{\sigma} \rho^2 \ket{\sigma}}{\bra{\sigma} \rho \ket{\sigma}^2}} = \sqrt{\kappa} \cdot  \eps~.
\end{align*}
\end{proof}

The reader may notice that the construction of $(\ket{C},\ket{D})$ in \Cref{lem:kappa-dependence} has Uhlmann fidelity $\fidelity(\rho,\sigma)$ that vanishes as the parameter $\kappa$ increases:
\[
    \frac{\bra{\sigma} \rho^2 \ket{\sigma}}{\bra{\sigma} \rho \ket{\sigma}^2} \leq \frac{\bra{\sigma} \rho \ket{\sigma}}{\bra{\sigma} \rho \ket{\sigma}^2} = \frac{1}{\bra{\sigma} \rho \ket{\sigma}} = \frac{1}{\fidelity(\rho,\sigma)^2}~.
\]
This is not inherent: the following simple modification of the construction yields a pair of states $(\ket{C},\ket{D})$ whose Uhlmann fidelity $\fidelity(\rho,\sigma)$ is at least $1/2$, but the spectral norm parameter $\kappa$ can be arbitrarily large, and furthermore the canonical Uhlmann transformation has rigidity that necessarily depends on $\kappa$.

Let $\ket{\tilde{C}}$ and $\ket{\tilde{D}}$ denote the states constructed in \Cref{lem:kappa-dependence}. Define new states
\[
    \ket{C} = \frac{1}{\sqrt{2}} \ket{\bot} + \frac{1}{\sqrt{2}} \ket{\tilde{C}}~, \qquad \ket{D} = \frac{1}{\sqrt{2}} \ket{\bot} + \frac{1}{\sqrt{2}} \ket{\tilde{D}}
\]
where $\ket{\bot}$ denotes an arbitary vector that is orthogonal to both $\ket{\tilde{C}}, \ket{\tilde{D}}$. Clearly the Uhlmann fidelity $\fidelity(\rho,\sigma)$ of $\ket{C},\ket{D}$ is at least $1/2$. On the other hand, one can calculate that the smallest nonzero eigenvalue of $\rho^{-1} \# \sigma$ is at least $1$, and the spectral norm parameter $\| \rho^{-1/2} P \rho^{1/2} \|^2_\infty$ is at least $1/\fidelity(\tilde{\rho},\tilde{\sigma})^2$ where $\fidelity(\tilde{\rho},\tilde{\sigma})$ is the Uhlmann fidelity of $\ket{\tilde{C}},\ket{\tilde{D}}$, which can be made arbitrarily large. Finally, the robustness of the canonical Uhlmann transformation for $\ket{C},\ket{D}$ obeys a similar dependence on $\kappa$ as in \Cref{lem:kappa-dependence}. 

\paragraph{A rounding lemma for the obliqueness?} A natural question is whether there is a rounding lemma for the obliqueness, similar to the rounding lemma for the spectral gap (\Cref{lem:closeness}). That is, could every pair of states be close to another pair with a controlled obliqueness? One could even ask for a stronger statement: is every pair of states close to a pair where \emph{both} the spectral gap and obliqueness are controlled? We leave these as interesting open questions. 

\section{Applications}
\label{sec:applications}

\subsection{The complexity of the Uhlmann Transformation Problem}
\label{sec:complexity}

Recently, Bostanci, Efron, Metger, Poremba, Qian, and Yuen~\cite{bostanci2023unitary} proposed a framework for studying the computational complexity of unitary transformations. A major focus of their study was on the Uhlmann Transformation Problem, which is the computational task of implementing canonical Uhlmann transformations corresponding to a specified pair of states. In particular, instances of the Uhlmann Transformation problem are pairs $(C,D)$ of circuit descriptions that act on two registers $\reg{AB}$, and the corresponding unitary transformation is the canonical Uhlmann transformation between $\ket{C} = C \ket{0}$ and $\ket{D} = D \ket{0}$. 

Bostanci, et al.~\cite{bostanci2023unitary} showed that the Uhlmann Transformation Problem, when restricted to instances whose reduced states on $\reg{A}$ have fidelity $1$, is complete for the unitary complexity class $\class{avgUnitaryHVPZK}$, the unitary complexity analogue of the decision class $\class{PZK}$ (perfect zero knowledge, i.e., problems with zero-knowledge protocols with a simulator that exactly reproduces the verifier's view) with an honest verifier.
For our discussion here, we need to recall the proof idea for the \emph{containment} of the Uhlmann transformation problem in $\class{avgUnitaryHVPZK}$.

The setup is as follows: a verifier is given classical descriptions of two circuits, $C$ and $D$, and wants to implement the Uhlmann transformation between the states $\ket{C}$ and $\ket{D}$ with the help of an all-powerful but untrusted prover.
The verifier and prover can exchange quantum messages.
To check that the untrusted prover applies the canonical Uhlmann transformation, the verifier prepares $\ket{C}$, sends $\reg{B}$ of $\ket{C}$ to the prover, and checks that the resulting state has high overlap with $\ket{D}$ (e.g.~using a SWAP test).  
If the verifier repeats this experiment many times and all of the experiments pass, by \cite[Proposition 6.3]{bostanci2023unitary} the verifier knows that the prover has implemented (a channel completion of) the canonical Uhlmann transformation.  
Furthermore, if the verifier wants to implement the canonical Uhlmann transformation on an arbitrary input state (instead of the $\reg B$-register of the state $C \ket 0$ prepared by the verifier itself), they can replace the copy of $\ket{C}$ in a random run of the experiment with a given input state.
Provided that the input state has the same distribution as the reduced state of $\ket{C}$ on the $\reg{B}$ register, \cite{bostanci2023unitary} shows that, up to an error that scales as the inverse of the number of experiments the verifier performs, the verifier will have implemented the canonical Uhlmann transformation on their input.
For their proof,~\cite{bostanci2023unitary} used the weak Uhlmann rigidity mentioned in \Cref{thm:weak-rigidity} to argue that the prover correctly mapping $\ket{C}$ to $\ket{D}$ with high fidelity must be approximately applying the canonical Uhlmann transformation.

When the fidelity of the reduced states on $\reg{A}$ is a constant $\gamma < 1$, however, the verifier only expects a $\gamma$ fraction of the trials to pass, even if the prover is honest and implements the canonical Uhlmann transformation.
By repeatedly sending copies of $\ket{C}$ and measuring $\proj{D}$, the verifier can estimate how well the prover is mapping $\ket{C}$ to $\ket{D}$, but it was not previously clear that the condition $\bra{D} \id \otimes R \ket{C} \geq \gamma- \epsilon$ implies that $R$ was close, in some notion, to the canonical Uhlmann transformation.  
Using \Cref{thm:main}, we can show that such a verifier does indeed implement a transformation close to the canonical Uhlmann transformation on their input state.  This yields a $2$-round quantum interactive protocol for synthesizing Uhlmann transformations with fidelity $\gamma$ (formally, solving the unitary synthesis problem $\avgUhlmann_{\gamma}$), similar to the $\class{avgUnitaryHVPZK}$ protocol presented in \cite{bostanci2023unitary}. Previously, the only way to synthesize $\avgUhlmann_{\gamma}$ without assuming a polarization conjecture was to go through that fact that the Uhlmann transformation is in $\class{avgUnitaryPSPACE} = \class{avgUnitaryQIP}$, yielding a $8$-round protocol.  

We note that while we are able to show that (roughly) the same algorithm works to synthesize $\avgUhlmann_{\gamma}$ instances, we are not able to show that $\avgUhlmann_{\gamma}$ is contained in $\class{avgUnitaryHVSZK}$.  In fact, the task of the $\class{avgUnitaryHVSZK}$ simulator for $\avgUhlmann_{\kappa}$ is $\class{QSZK}$-hard in general, as it allows one to estimate the fidelity between the input states $\ket{C}$ and $\ket{D}$.  Thus, without a major complexity theoretic breakthrough, $\avgUhlmann_{\gamma}$ (for constant $\gamma$) is unlikely to be shown to be in $\class{avgUnitaryHVSZK}$.  We also note that our algorithm is only efficient for so-called ``well-conditioned'' instances of $\avgUhlmann_{\gamma}$, namely those whose spectral gap and obliqueness are polynomial in $n$ and $r$.  While we do not formalize the notion of ``well-conditioned'' Uhlmann instances, it is not hard to see how it could be formalized from the protocol below.

\begin{longfbox}[breakable=false, padding=1em, margin-top=1em, margin-bottom=1em]
    \begin{protocol} {\bf $2$-round quantum interactive protocol for $\avgUhlmann_{\gamma}$ } \label{prot:uhlmann_two_round}
    \end{protocol}
    \noindent \textbf{Instance: } A valid $\avgUhlmann_{\gamma}$ instance $x = (1^{n}, C, D)$, and precision $r \in \mathbb{N}$, with spectral gap $\eta$ and obliqueness $\kappa$\\
    \noindent \textbf{Input: } An $n$ qubit quantum register $\reg{B}_0$.  
    \begin{enumerate}
        \item Let $m = 8n(\kappa r / \eta)^2$.  Sample $i^* \in [m]$ uniformly at random.  Initialize $j \leftarrow 0$.
        \item For $i = 1$ through $m$:
        \begin{enumerate}
            \item If $i \neq i^*$:
            \begin{enumerate}
                \item Prepare the state $\ket{C}_{\reg{A'}\reg{B}'}$ and send $\reg{B}'$ to the prover.
                \item After receiving $\reg{B}'$ from the prover, apply $D^{*}$ to $\reg{A}'\reg{B}'$ and measure all of the qubits in the computational basis.  \label[step]{step:Uhlmann_verifier_measurement}
                \item If the measurement outcome is $0^{n}$, increment $j$ by $1$.
            \end{enumerate}
            \item If $i = i^*$:
            \begin{enumerate}
                \item Send $\reg{B}_0$ to the prover and receive $\reg{B}_0$ back.
            \end{enumerate}
        \end{enumerate}
        \item If $j \geq m \cdot \left(\gamma - \frac{\eta}{4 \kappa r}\right)$, accept and output $\reg{B}_0$, otherwise reject.
    \end{enumerate}
\end{longfbox}

We prove that the protocol satisfies both the completeness and soundness conditions in Definition 4.9 of \cite{bostanci2023unitary}.
Here we will make use of the terminology and notation from~\cite{bostanci2023unitary}.

\begin{lemma}[Completeness]
\label{lem:qip_completeness}
    For all valid $\avgUhlmann_{\gamma}$ instances $(1^{n}, C, D)$ and error parameter $r \in \mathbb{N}$, for sufficiently large $n$ and $r$, the honest prover that applies the canonical Uhlmann transformation between $\ket{C}$ and $\ket{D}$ is accepted with probability at least $1 - 2^{-n}$.
\end{lemma}
\begin{proof}
    Since the honest prover applies the canonical Uhlmann transformation, every measurement in \Cref{step:Uhlmann_verifier_measurement} accepts with probability $\kappa$ independently.  Let $X_{i}$ be the event that the $i$'th measurement accepts, the probability that the verifier rejects the honest prover is upper bounded by Hoeffding's inequality as
    \begin{align*}
        \Pr\left[\sum_{i} X_i \leq m \cdot \left(\gamma - \frac{\eta}{4\kappa r}\right)\right] &\leq \Pr\left[\left|\sum_{i} X_i - \mathbb{E}\left[\sum_{i} X_i\right]\right| \geq \frac{m}{(2\kappa r / \eta)}\right]\\
        &\leq \exp\left(-2m / (4\kappa r / \eta)^2\right)\\
        &= 2\exp(-n)\\
        &\leq 2^{-n}\,.\qedhere
    \end{align*}
\end{proof}

\begin{lemma}[Soundness]
\label{lem:qip_soundness}
    For all valid $\avgUhlmann_{\kappa}$ instances $(1^{n}, C, D)$, and error parameters $r \in \mathbb{N}$, for sufficiently large $n$ and $r$, for all quantum provers $P$, there exists a channel completion $\Phi_{x}$ of the canonical Uhlmann transformation $U_{x}$ such that
    \begin{equation*}
        \mathrm{if}\quad \Pr[V_{x, r}(\ket{C}) \interact P\ \mathrm{accepts}] \geq \frac{1}{2}\quad \mathrm{then}\quad \td(\sigma, (\Phi_{x}\otimes \id)(\proj{C})) \leq \frac{1}{r}\,,
    \end{equation*}
    where $\sigma$ denotes the output of $V_{x, r}(\ket{C}) \interact P$, conditioned on $V_{x, r}$ accepting.  
\end{lemma}

\begin{proof}
    Similar to \cite{bostanci2023unitary}, we can write the state of the verifiers register, before they do the measurements in \Cref{step:Uhlmann_verifier_measurement} as 
    \begin{equation*}
        \rho = (\Lambda \otimes \id)(\proj{C}^{\otimes m})\,,
    \end{equation*}
    where $\Lambda$ denotes the channel that the prover applies to registers $\reg{B}_1\ldots \reg{B}_m$, and importantly the state $\rho$ is permutation-invariant.  We can imagine the process of sampling a $m$-bit string $X$ by measuring each register $\reg{A}_i\reg{B}_i$ using the POVM $\{\proj{D}, \id - \proj{D}\}$ and writing down the answer as $X_{i}$.  For a subset $S \subseteq [m]$ of size $\left(\kappa - \frac{1}{r}\right)m$, let the event $E_{S}$ be the event that all of the bits $X_i$ with $i \in S$ are $1$, then we have the following
    \begin{align*}
        &\Pr\left[X_{i^*} = 0\ \Big|\ \sum_{i \neq i^*} X_{i} \geq \left(\gamma - \frac{\eta}{4\kappa r}\right)m\right]\\
        &\hspace{10mm}= \frac{1}{m} \sum_{i^*} \mathop{\mathbb{E}}_{|S| = \left(\gamma - \frac{1}{r}\right)m} \left[\Pr\left[X_{i^*} = 0 | X_i = 1\ \forall i \in S\right]\right]\\
        &\hspace{10mm}= \left(1 - \gamma + \frac{\eta}{4\kappa r}\right) + \frac{1}{m} \sum_{i^*} \mathop{\mathbb{E}}_{\substack{|S| = \left(\gamma - \frac{1}{r}\right)m\\ i^* \not\in S}} \left[\Pr\left[X_{i^*} = 0 | X_i = 1 \ \forall i \in S\right]\right]\\
        &\hspace{10mm}\leq \left(1 - \gamma + \frac{\eta}{4\kappa r}\right) + \frac{2}{m} \sum_{i^*}\mathop{\mathbb{E}}_{\substack{|S| = \left(\gamma - \frac{\eta}{4\kappa r}\right)m\\ i^* \not\in S}} \left[\Pr\left[X_{i^*} = 0 \land X_i = 1 \ \forall i \in S\right]\right]\\
        &\hspace{10mm}= \left(1 - \gamma + \frac{\eta}{4\kappa r}\right) + \frac{2}{m} \mathop{\mathbb{E}}_{|S| = \left(\gamma - \frac{1}{r}\right)m} \left[\sum_{i^* \not\in S} \Pr\left[X_{i^*} = 0 \land X_i = 1 \ \forall i \neq i^*\right]\right]\\
        &\hspace{10mm}\leq 1 - \gamma + \frac{\eta}{4\kappa r} + \frac{2}{m}\\
        &\hspace{10mm}\leq 1 - \gamma + \frac{\eta}{2 \kappa r}\,.
    \end{align*}
    Here we use the fact that $\rho$ is permutation-invariant to average over all permutations of the registers $\reg{A}_i\reg{B}_i$, then we use the fact that the probability of the verifier measuring a greater than $\gamma - \frac{\eta}{4\kappa r}$ fraction of accepts is at least $1/2$, and finally we use the fact that when restricting to the variables in a fixed $S$, all of the events $X_{i^*} = 0 \land X_i = 1\ \forall i \neq i^*$ are mutually exclusive, so their probabilities sum to $1$.  Finally we use the definition of $m$ to upper bound $2/m$.  

    Thus, we have a bound on the probability of measuring the $i^*$'th register in the state $\proj{D}$.  Put another way, there exist a purification of the provers channel, $R$, and states $\ket{\psi}$ and $\ket{\phi}$ such that
    \begin{equation*}
        \left|\bra{D}\bra{\phi} (\id \otimes R) \ket{C}\ket{\psi}\right|^2 \geq \gamma - \frac{\eta}{2\kappa r}\,.
    \end{equation*}
    Thus, by \Cref{thm:main}, we have that
    \begin{equation*}
        \|\id \otimes (W - R) \left(W^*W \otimes \proj{\psi}\right)\ket{C}\ket{\psi}\|\leq \frac{1}{r}\,.
    \end{equation*}
    Tracing out the registers used to purify the prover, we see that the verifier satisfies the soundness condition of a quantum interactive protocol for a distributional unitary synthesis problem.
\end{proof}

\begin{theorem}[$2$-round quantum interactive protocol for $\avgUhlmann$]
    For all $\gamma$, $\eta = 1/\poly(n)$ and $\kappa = \poly(r)$, there is a $2$-round $\class{avgUnitaryQIP}$ protocol with completeness error $2^{-n}$ and soundness $\frac{1}{2}$ for synthesizing instances of $\avgUhlmann_{\gamma}$ with spectral gap $\eta$ and obliqueness $\kappa$. 
\end{theorem}
\begin{proof}
    Follows immediately from \Cref{lem:qip_completeness}, \Cref{lem:qip_soundness}, and the definition of \Cref{prot:uhlmann_two_round}.
\end{proof}

\subsection{Stability of approximate representations}
\label{sec:gowers-hatami}

To demonstrate the flexibility and power of the robust rigidity theorem for Uhlmann transformations, we show how a stability theorem for approximate representations of groups arises from the robust rigidity of Uhlmann transformations.  The proof draws inspiration from the proof of Gowers-Hatami theorem with non-uniform measures from \cite{metger2024succinct}, and similarly achieves a worse dimension blowup than the result of \cite{gowers2015inverse}. However, we believe that presenting the proof this way gives a framework for proving similar stability theorems, and highlights how the robust rigidity of Uhlmann transformations is a very general kind of stability theorem.  

As a matter of notation, for a probability measure $\mu$ over a finite set $S$, we use the notation $g \sim \mu$ to describe sampling $g$ from the distribution $\mu$, and $h \sim S$ to denote $h$ sampled uniformly at random from $S$.

\begin{definition}[$\epsilon$-approximate representation]
\label{def:approximate_representation}
    Let $G$ be a finite group, $\mu$ be a measure over $G$, and $\rho$ be a quantum state on register $\reg{B}$. A collection of unitaries $\{U_{g}\}_{g \in G}$ acting on $\reg{B}$ is an $\epsilon$-approximate representation of $G$ on $\ket{\psi}$ with measure $\mu$ if 
    \begin{equation*}
        \mathop{\mathbb{E}}_{\substack{g\sim \mu\\ h \sim G}}\left[\|U_{h}U_{g} - U_{hg}\|^2_{\rho}\right] \leq \epsilon\,,
    \end{equation*}
    where $\|A\|_{\rho} = \sqrt{\Tr(A^* A \rho)}$.
\end{definition}

\begin{remark}
\label{rem:epsilon_approximate_inner_product}
    The condition that $\{U_{g}\}_{g \in G}$ is an $\epsilon$-approximate representation on the state $\rho$ implies that for all purifications $\ket{\psi}$ of $\rho$, the following holds:
    \begin{equation*}
        \mathop{\mathbb{E}}_{\substack{g\sim \mu\\ h \sim G}} \left[\mathrm{Re}\left(\bra{\psi} \left(\id \otimes U_{hg}^{*} U_{h}U_{g}\right)\ket{\psi}\right)\right] \geq 1 - \epsilon / 2\,.
    \end{equation*}
\end{remark}

Our goal will be to prove that $\epsilon$-approximate representations (in the sense of \Cref{def:approximate_representation}) are close to exact representations, up to an isometry.  
Formally, we will prove the following theorem.
\begin{theorem}[Stability of approximate representations]
\label{thm:stability_approximate_representations}
    Let $G$ be a finite group, $\mu$ be a measure over $G$, and $\rho$ be a bipartite state on registers $\reg{B}$.  Let $\{U_{g}\}_{g \in G}$ be an $\epsilon$-approximate representation of $G$ on $\rho$ over $\mu$.
    Then there exists an exact representation of $G$, $\mathcal{R}$, and isometry $V$ such that
    \begin{equation*}
        \mathop{\mathbb{E}}_{g \sim \mu}\left[\|U_{g} - V^{*} \mathcal{R}(g) V\|_{\rho}^2\right] \leq \epsilon\,.
    \end{equation*}
\end{theorem}
To prove the theorem, we will exhibit a pair of states $\ket{C}$ and $\ket{D}$ such that the rigidity of the Uhlmann transformation between them implies the stability theorem.  Let $\{U_{g}\}_{g \in G}$ be an approximate representation of a group $G$ in the sense of \Cref{def:approximate_representation}, and let $\ket{\psi}$ be a purification of $\rho$ (for example $\id \otimes \sqrt{\rho} \ket{\Omega}$, although any purification will suffice).  Then consider the following pair of states:
\begin{align*}
    \ket{C} &= \frac{1}{\sqrt{|G|}}\sum_{g, h \in G} \sqrt{\mu(g)}\left(\id \otimes U_{g}\right)\ket{\psi}_{\reg{AB}_1}\ket{g}_{\reg{B}_2}\ket{h}_{\reg{B}_3}\,,\\
    \ket{D} &= \frac{1}{\sqrt{|G|}}\sum_{g, h \in G} \sqrt{\mu(g)} \left(\id \otimes U_{hg}\right)\ket{\psi}_{\reg{AB}_1}\ket{g}_{\reg{B}_2}\ket{h}_{\reg{B}_3}\,.
\end{align*}
To prove the stability theorem, we prove the following lemmas.

\begin{lemma}[Uhlmann transformation for $\ket{C}$ and $\ket{D}$]
\label{lem:canonical_uhlmann_representation}
    The following transformation between $\ket{C}$ and $\ket{D}$ achieves the Uhlmann fidelity. 
    \begin{align*}
        \widetilde{W} = \sum_{g, h} \left(U_{hg} U_{g}^{*}\right)_{\reg{B}_1} \otimes \ket{g, h}\!\!\bra{g, h}_{\reg{B}_2\reg{B}_3}\,.
    \end{align*}
\end{lemma}

Note that from Uhlmann's theorem, the reduced states on register $\reg A$ of $\ket{C}$ and $\ket{D}$ are identical.  Furthermore, since the reduced states are the same, they commute with each other, and $\rho^{-1} \# \sigma$ is the projector onto their positive eigenspace, and all of its non-zero eigenvalues are $1$.  Therefore, both $\eta$ and $\kappa$ from \Cref{thm:main} are $1$.  Further, if we let $W$ be the \emph{canonical} Uhlmann transformation\footnote{
In general, $\widetilde{W}$ will not be the canonical Uhlmann transformation, since $\mu$ might not assign non-zero probability to all $g$.  However in the case when $\mu$ is the uniform distribution, $\widetilde{W}$ is the canonical Uhlmann transformation between $\ket{C}$ and $\ket{D}$.  
}, it is clear that $W^* W \ket{C} = \ket{C}$, because $W^* W$ is a projector and $W$ maps $\ket{C}$ to $\ket{D}$, implying $\ket{C}$ is in the image if $W$.  

Next, we show that the approximate representation corresponds to an approximate Uhlmann transformation between $\ket{C}$ and $\ket{D}$.  

\begin{lemma}
\label{lem:approximate_representation_is_uhlmann}
    Let $\{U_{g}\}$ is an $\epsilon$-approximate representation of $G$ under $\ket{\psi}$.  Consider the following unitary
    \begin{equation*}
        U = \sum_{h} \left(U_{h}\right)_{\reg{B}_1} \otimes \proj{h}_{\reg{B}_3}\,.
    \end{equation*}
    Then $U$ achieves the following.
    \begin{equation*}
        |\bra{D} \id \otimes U \ket{C}| \geq 1 - \frac{\epsilon}{2}\,.
    \end{equation*}
\end{lemma}

Finally, we show that the implication of the rigid robustness theorem corresponds exactly to being close to a specific exact representation of $G$.  

\begin{lemma}
\label{lem:rigidity_implies_stability}
    Define the following representation $\mathcal{R}$ of $G$ and an isometry $V$:
    \begin{equation*}
        \mathcal{R}(g) = \sum_{h} \ket{h}\!\!\bra{hg} \quad\text{and}\quad V = \frac{1}{\sqrt{|G|}}\sum_{h} \left(U_{h}\right)_{\reg{B}_1} \otimes \ket{h}_{\reg{B}_3}\,.
    \end{equation*}
    Then the following holds: 
    \begin{equation*}
        \|\id \otimes (U - \widetilde{W})\ket{C}\|^2 = \mathop{\mathbb{E}}_{g \sim \mu} \left[\|U_g - V^{*} \mathcal{R}(g) V\|_{\rho}^2\right]\,.
    \end{equation*}
\end{lemma}

Combining all of these, we have the following simple proof of \Cref{thm:stability_approximate_representations}.
\begin{proof}[Proof of \Cref{thm:stability_approximate_representations}]
    From \Cref{lem:approximate_representation_is_uhlmann}, \Cref{lem:canonical_uhlmann_representation}, and \Cref{thm:main}, together with the fact that $\eta = \kappa = 1$ in \Cref{thm:main} for the states $\ket{C}$ and $\ket{D}$, we have that 
    \begin{equation*}
        \|\id \otimes (U - \widetilde{W})\ket{C}\|^2 = \|\id \otimes (U - W) W^* W \ket{C}\|^2 \leq \epsilon\,.
    \end{equation*}
    Here, $W$ is the canonical Uhlmann transformation, and we use the fact that, from the rigidity theorem, $\widetilde{W}$ is equal to $W$ when restricted to the image of $W$.  Then from \Cref{lem:rigidity_implies_stability}, we have that
    \begin{equation*}
        \mathop{\mathbb{E}}_{g \sim \mu} \left[\|U_g - V^{*} \mathcal{R}(g) V\|_{\rho}^2\right] \leq \epsilon\,,
    \end{equation*}
    as desired.
\end{proof}

We now fill in the missing proofs.

\begin{proof}[Proof of \Cref{lem:canonical_uhlmann_representation}]
    To prove the lemma, we apply the definition of $\widetilde{W}$. 
    \begin{align*}
        \widetilde{W}\ket{C} &= \widetilde{W} \left(\frac{1}{\sqrt{|G|}} \sum_{g, h \in G} \sqrt{\mu(g)}(\id \otimes U_{g}) \ket{\psi}_{\reg{AB}_1}\ket{g}_{\reg{B}_2}\ket{h}_{\reg{B}_3}\right)\\
        &= \frac{1}{\sqrt{|G|}} \sum_{g, h \in G} \sqrt{\mu(g)} (\id \otimes U_{hg} U_{g}^{*}U_{g}) \ket{\psi}_{\reg{AB}_1}\ket{g}_{\reg{B}_2}\ket{h}_{\reg{B}_3}\\
        &= \frac{1}{\sqrt{|G|}} \sum_{g, h \in G} \sqrt{\mu(g)} (\id \otimes U_{hg}) \ket{\psi}_{\reg{AB}_1}\ket{g}_{\reg{B}_2}\ket{h}_{\reg{B}_3}\\
        &= \ket{D}\,.\qedhere
    \end{align*}  
\end{proof}

\begin{proof}[Proof of \Cref{lem:approximate_representation_is_uhlmann}]
    We can expand out the real component of the inner product between $\ket{D}$ and $U$ applied to $\ket{C}$ as follows. 
    \begin{align*}
        \mathrm{Re}\left(\bra{D} \id \otimes U \ket{C}\right) &= \frac{1}{|G|}\sum_{g, h \in G} \mu(g) \mathrm{Re}\left(\bra{\psi} (\id \otimes U_{hg}^{*} U_{h}U_{g})\ket{\psi}\right)\\
        &= \mathop{\mathbb{E}}_{\substack{g \sim \mu\\ h \in G}}\left[\mathrm{Re}\left(\bra{\psi} (\id \otimes U_{hg}^{*} U_{h}U_{g})\ket{\psi}\right)\right]\\
        &\geq 1 - \frac{\epsilon}{2}\,.
    \end{align*}
    Here in the final line we use \Cref{rem:epsilon_approximate_inner_product} and the fact that $\{U_{g}\}_{g \in G}$ is an $\epsilon$-approximate representation under $\ket{\psi}$.  Since $|\bra{D} \id \otimes U \ket{C}| \geq \mathrm{Re}(\bra{D} \id \otimes U \ket{C})$, this completes the proof.
\end{proof}

\begin{proof}[Proof of \Cref{lem:rigidity_implies_stability}]
    We first note that
    \begin{align*}
        V^{*} \mathcal{R}(g) V &= \frac{1}{|G|}\sum_{h, h'} \left(U^{*}_{h} \otimes \bra{h}\right) \mathcal{R}(g) \left(U_{h'} \otimes \ket{h'}\right)\\
        &= \frac{1}{|G|} \sum_{h, h'} \left(U^{*}_{h} \otimes \bra{hg}\right) \left(U_{h'} \otimes \ket{h'}\right)\\
        &= \frac{1}{|G|} \sum_{h} U^{*}_{h} \cdot U_{hg}\,.
    \end{align*}
    We can interpret this as the convolution of the function $f(g) = U_{g}$ with itself.  Here in the second to last line, we re-index the sum over $h'$ by shifting every element by $g$.  Note that if $U_{g}$ was an exact representation of $G$, then we would have $U^{*}_{h}U_{hg} = U_{h}^{*} U_{h} U_{g} = U_{g}$.  
    
    Now to prove the lemma, we expand the definition of the $2$-norm to get the following
    \begin{align*}
        &\|\id \otimes (U - \widetilde{W}) \ket{C}\|^{2} \\
        &\hspace{10mm}=\bra{C} \left(\id \otimes \left((U - \widetilde{W})^{*} (U - \widetilde{W}) \right)\right) \ket{C}\\
        &\hspace{10mm}= 2 - 2 \mathrm{Re}\left(\bra{C} \left(\id \otimes U^{*} \widetilde{W}\right) \ket{C}\right)\\
        &\hspace{10mm}= 2 - \frac{2}{|G|} \sum_{g, g', h, h'} \sqrt{\mu(g)\mu(g')}\mathrm{Re}\left(\bra{\psi} \bra{g, h} \left(U_{g}^{*} \cdot U^{*} \cdot \widetilde{W} \cdot U_{g'} \right) \ket{\psi} \ket{g', h'}\right)\\
        &\hspace{10mm}= 2 - \frac{2}{|G|} \sum_{g, g', h, h'} \sqrt{\mu(g)\mu(g')}\mathrm{Re}\left(\bra{\psi} \bra{g, h} \left(U_{g}^{*} \cdot U_{h}^{*} \cdot U_{h'g'} \right) \ket{\psi} \ket{g', h'}\right)\\
        &\hspace{10mm}= 2 - \frac{2}{|G|}\sum_{g, h}\mu(g)\mathrm{Re}\left(\bra{\psi} U^{*}_{g} \cdot U_{h}^{*} \cdot U_{hg} \ket{\psi}\right)\\
        &\hspace{10mm}= 2 - 2 \sum_{g}\mu(g)\mathrm{Re}\left(\bra{\psi} U_{g}^{*} \cdot \left(\frac{1}{|G|} \sum_{h}U^{*}_{h} \cdot U_{hg}\right) \ket{\psi}\right)\\
        &\hspace{10mm}= 2 - 2 \sum_{g} \mu(g)\mathrm{Re}\left(\bra{\psi} U_{g}^{*} \cdot \left(V^{*} \mathcal{R}(g) V\right) \ket{\psi}\right)\\
        &\hspace{10mm}= \mathbb{E}_{g \sim \mu} \left[\|U_{g} - V^{*} \mathcal{R}(g) V\|_\rho^2\right]\,.
    \end{align*}
    Here we apply the definition on $\widetilde{W}$ and using the fact that $\mathrm{Re}$ is linear to move the sum over $h$ inside, and using the definition to replacing $\sum_{h} U^{*}_{h} U_{hg}$ with $V^{*} \mathcal{R}(g) V$.  
\end{proof}

\section*{Acknowledgments} 

JB and HY are supported by AFOSR award FA9550-23-1-0363, NSF CAREER award CCF-2144219, and the Sloan Foundation. TM acknowledges support from SNSF Grant No. 20QU-1\_225171 and NCCR SwissMAP.

\printbibliography
\end{document}